\documentclass{article}
\usepackage{geometry}
\usepackage[utf8]{inputenc}
\usepackage{cite}
\usepackage{url}
\usepackage{amssymb,amsmath,amsthm}
\usepackage{bm}
\usepackage{graphicx}
\usepackage{enumerate}
\usepackage{microtype}
\usepackage{color}
\usepackage{hyperref}
\usepackage{multirow}
\usepackage{array}
\newcolumntype{C}[1]{>{\centering\let\newline\\\arraybackslash\hspace{0pt}}m{#1}}

\usepackage{algorithm}
\usepackage{algpseudocode}
\usepackage{amsmath}

\newtheorem{defn}{\noindent $\mathbf{Definition}$}[section]

\newtheorem{proposition}[defn]{$\mathbf{Proposition}$}

\title{Bijective Density-Equalizing Quasiconformal Map for Multiply-Connected Open Surfaces}

\author{Zhiyuan Lyu\thanks{Department of Mathematics, The Chinese University of Hong Kong
  ({zylyu@math.cuhk.edu.hk}).}
\and Gary P. T. Choi\thanks{Department of Mathematics, The Chinese University of Hong Kong
  ({ptchoi@math.cuhk.edu.hk}).}
  \and Lok Ming Lui\thanks{Department of Mathematics, The Chinese University of Hong Kong
  ({lmlui@math.cuhk.edu.hk}).}}

\date{}

\begin{document}

\maketitle
\begin{abstract}
This paper proposes a novel method for computing bijective density-equalizing quasiconformal (DEQ) flattening maps for multiply-connected open surfaces. In conventional density-equalizing maps, shape deformations are solely driven by prescribed constraints on the density distribution, defined as the population per unit area, while the bijectivity and local geometric distortions of the mappings are uncontrolled. Also, prior methods have primarily focused on simply-connected open surfaces but not surfaces with more complicated topologies. Our proposed method overcomes these issues by formulating the density diffusion process as a quasiconformal flow, which allows us to effectively control the local geometric distortion and guarantee the bijectivity of the mapping by solving an energy minimization problem involving the Beltrami coefficient of the mapping. To achieve an optimal parameterization of multiply-connected surfaces, we develop an iterative scheme that optimizes both the shape of the target planar circular domain and the density-equalizing quasiconformal map onto it. In addition, landmark constraints can be incorporated into our proposed method for consistent feature alignment. The method can also be naturally applied to simply-connected open surfaces. By changing the prescribed population, a large variety of surface flattening maps with different desired properties can be achieved. The method is tested on both synthetic and real examples, demonstrating its efficacy in various applications in computer graphics and medical imaging.
\end{abstract}

\section{Introduction}\label{sec:intro}
In recent years, three-dimensional (3D) data has become increasingly ubiquitous in various fields, including computer graphics, computer vision, and medical imaging. In computer graphics, 3D models are used to create realistic animations, games, and virtual environments. In computer vision, 3D data is used for 3D scene understanding, object recognition and tracking. In medical imaging, 3D data is used to create detailed models of organs and tissues, aiding in diagnosis and treatment planning. However, processing and manipulating 3D data can be challenging due to its complex and irregular nature. One important technique for addressing this challenge is surface parameterization, which involves mapping a surface in $\mathbb{R}^3$ onto a simpler parameter domain while preserving certain desired properties. Surface parameterizations have a wide range of applications in processing 3D data. For example, remeshing involves generating a new mesh with desired properties from an existing mesh, which can be achieved more easily on a parameterized surface. Visualization of 3D models is also facilitated by parameterization, as it allows for efficient rendering and texture mapping. One can also solve partial differential equations (PDEs) on a surface by parameterizing it onto a simpler two-dimensional (2D) domain where the PDEs can be solved more easily. Once the solution is obtained on the 2D domain, it can be mapped back to the original surface using the inverse of the parameterization map.

Different surface parameterization methods have been proposed depending on the properties of the surface parameterizations required for a specific application. Conformal parameterizations, for example, preserve angles and shapes locally and are commonly used in texture mapping and geometric processing. By contrast, area-preserving parameterizations preserve areas and are useful in physical simulations where mass is conserved. Another type of surface parameterizations is the distance-preserving parameterization, which preserves distances between points on the surface. 

Density-equalizing cartograms have recently gained attention in the field of surface parameterizations as a type of map projection technique. In~\cite{gastner2004diffusion}, Gastner and Newman proposed a novel algorithm for generating area cartograms based on the principle of density diffusion. The method deforms a given map such that the density of a prescribed quantity, defined by the population per unit area, becomes constant throughout the deformed map. The technique has found widespread use in data visualization, particularly in the visualization of sociological data such as global population, income, and age-of-death in different regions~\cite{dorling2008atlas}. Recently, Choi and Rycroft~\cite{choi2018density} proposed a finite-element approach for computing density-equalizing flattening maps of simply-connected open surfaces in $\mathbb{R}^3$ onto a parameter domain in $\mathbb{R}^2$, which has later been extended and applied to medical imaging~\cite{choi2020area} and 3D microstructure morphology~\cite{shaqfa2022geometrical,shaqfa2023disk}. However, there remains a gap in the literature with respect to the computation of density-equalizing flattening maps for multiply-connected open surfaces. Additionally, current approaches for computing density-equalizing parameterizations are limited in their ability to control local geometric distortion and bijectivity. These challenges pose significant obstacles to the applicability of these approaches in practical applications.

\begin{figure}[t]
    \centering
    \includegraphics[width=0.95\textwidth]{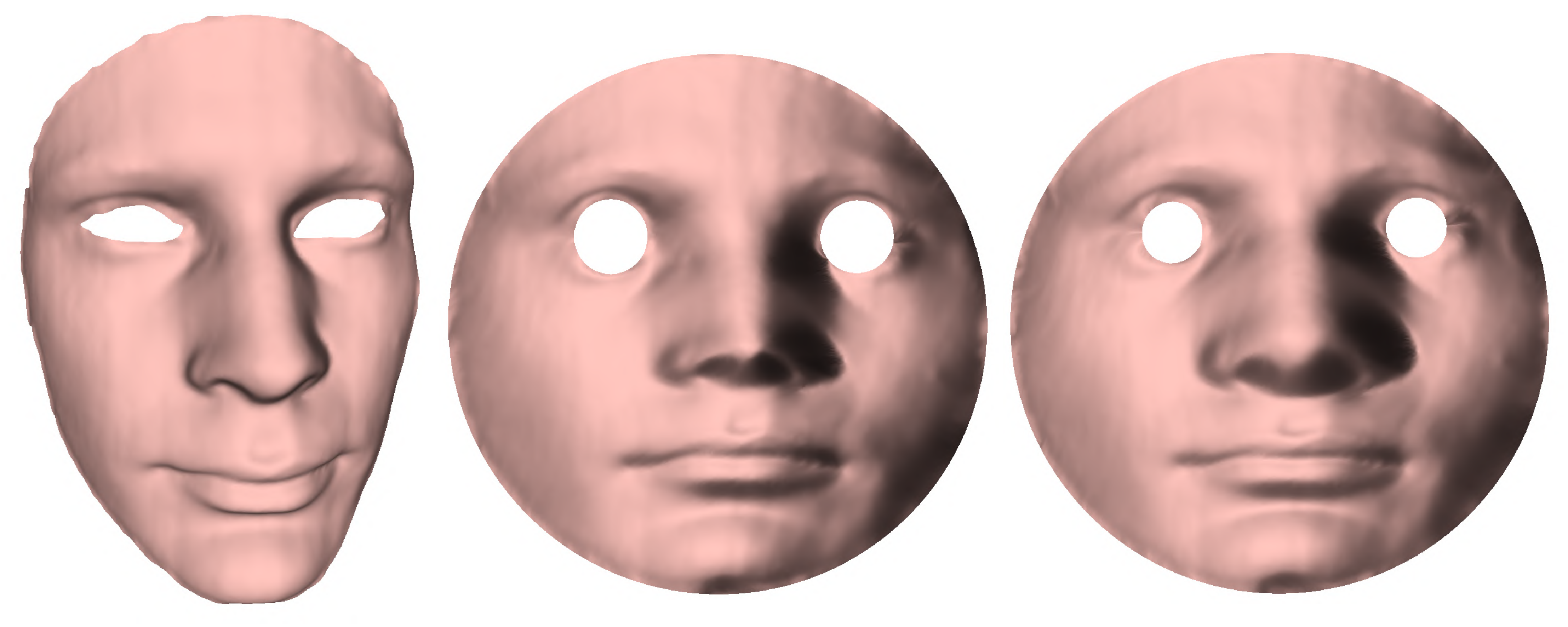}
    \caption{Bijective density-equalizing quasiconformal (DEQ) maps obtained by our proposed algorithm. Given a multiply-connected open surface and a prescribed density distribution, our method optimizes both the shape of the target 2D circular domain and the flattening map of the surface onto the domain to achieve density equalization, with the bijectivity of the mapping guaranteed by quasiconformal theory. By changing the prescribed density distribution, different effects can be achieved. For instance, we can parameterize a human face (left) onto the plane with the nose shrunk (middle) or enlarged (right) by changing the prescribed density distribution at the nose. Note that the size and position of the inner circular holes are also optimized and hence different in the two parameterization results.}
    \label{fig:overview}
\end{figure}

Motivated by this, we propose in this paper a novel method for computing bijective density-equalizing quasiconformal (DEQ) flattening maps for multiply-connected open surfaces, with controlled local geometric distortions and prescribed constraints on the density distribution (see Fig.~\ref{fig:overview}). Our proposed method utilizes quasiconformal flows to control the local geometric distortion during the diffusion process of the density distribution. Specifically, by formulating the diffusion process as a flow of the Beltrami coefficient, which measures the quasiconformality of the map, we can effectively control the local geometric distortion of the map by minimizing an energy functional involving the Beltrami coefficient terms. The bijectivity of the map is guaranteed by a projection operator that truncates the magnitude of the Beltrami coefficient. Our algorithm is capable of generating a large variety of flattening maps with different desired properties for various real-world applications by changing the initial population distribution. Moreover, the proposed method can flatten a multiply-connected open surface to a 2D circular domain, i.e. a unit disk with inner circular holes, where the size and location of the inner holes are further optimized to achieve quasiconformality and density equalization across the entire target domain. Additionally, our algorithm can accommodate landmark constraints to ensure consistent feature alignment. The method is also naturally applicable to simply-connected open surfaces. The effectiveness of our proposed algorithm has been verified through testing on both synthetic and real-world examples, showcasing its potential in numerous applications within the fields of computer graphics and medical imaging. 

The contributions of this paper can be summarized as follows:
\begin{enumerate}
\item We present an algorithm for computing density-equalizing surface parameterizations that overcomes the limitations of the conventional approaches in terms of bijectivity and quasiconformality.
\item To the best of our knowledge, this is the first work addressing the density-equalizing mapping problem as a quasiconformal flow problem, which allows us to incorporate quasiconformal techniques to minimize local geometric distortion in the density-equalizing parameterization. 
\item Our algorithm is capable of computing a bijective density-equalizing quasiconformal parameterization for multiply-connected open surfaces, which cannot be handled by conventional approaches. It can also be naturally applied to simply-connected open surfaces.
\item Our algorithm optimizes the geometry of the target 2D circular domain, i.e. the centers and radii of the inner circles in the circular domain, to obtain a parameterization with reduced geometric distortion for multiply-connected surfaces.
\item Our algorithm is the first density-equalizing mapping approach that incorporates feature landmark alignments, resulting in a bijective landmark-matching density-equalizing quasiconformal map.
\item Our proposed algorithm can be applied to various practical applications including texture mapping, surface remeshing and medical visualization, demonstrating the efficacy and practicality of our method.
\end{enumerate}

The organization of the paper is as follows. In Section \ref{sec:related}, we review the previous works in surface parameterization and density-equalizing maps. In Section \ref{sec:math}, we introduce the mathematical concepts related to our work. In Section~\ref{sec: Proposed math module}, we present our proposed mathematical models for computing bijective density-equalizing quasiconformal maps. In Section \ref{sec:main}, we describe the details of our proposed algorithms. Experimental results are presented in Section \ref{sec:experiments} for demonstrating the effectiveness of our proposed method. In Section \ref{sec:application}, we present applications of our method in different fields. We conclude our work and discuss future directions in Section \ref{sec:conclusion}.

\section{Related work}\label{sec:related}
Over the past few decades, surface parameterization has been widely studied in the area of geometry processing, graphics, and computer vision (see~\cite{floater2005surface,hormann2007mesh, sheffer2007mesh} for surveys). In particular, since isometric (both angle-preserving and area-preserving) parameterization is generally impossible to achieve for non-developable surfaces, a large variety of surface parameterization algorithms have been proposed for minimizing either the angle distortion (conformal parameterization) or the area distortion (area-preserving parameterization).  

Conformal parameterization~\cite{gu2008computational} preserves the local geometry of surfaces but leads to area distortion. Existing conformal parameterization methods for simply-connected open surfaces include least-squares conformal map (LSCM)~\cite{levy2002least}, holomorphic 1-form~\cite{gu2003global}, Yamabe flow~\cite{luo2004combinatorial,wang2007brain}, spectral conformal mapping~\cite{mullen2008spectral}, fast disk conformal map~\cite{choi2015fast}, conformal energy minimization~\cite{yueh2017efficient}, linear disk conformal map~\cite{choi2018linear}, and parallelizable global conformal parameterization~\cite{choi2020parallelizable,choi2022free}. Compared with simply-connected open surfaces, the conformal mapping of multiply-connected surfaces is less studied. Existing conformal parameterization methods for multiply-connected surfaces include the Koebe's iteration method~\cite{koebe1910konforme}, generalized Koebe's method~\cite{zeng2009generalized}, Laurent series~\cite{kropf2014conformal}, discrete conformal equivalence~\cite{bobenko2016discrete}, and poly-annulus conformal map (PACM)~\cite{choi2021efficient}. 

Quasiconformal maps are a generalization of conformal maps. In recent years, there has been an increasing interest in quasiconformal surface parameterization methods~\cite{choi2023recent}. Existing methods for computing quasiconformal parameterization include the auxiliary metric~\cite{zeng2009surface}, Beltrami holomorphic flow (BHF)~\cite{lui2012optimization}, bounded distortion parameterization~\cite{aigerman2014lifted, lipman2012bounded}, Linear Beltrami solver (LBS)~\cite{lui2013texture}, quasi-conformal energy minimization (QCMC)~\cite{ho2016qcmc}, fast spherical quasiconformal parameterization~\cite{choi2016fast}, least-squares quasi-conformal map (LSQC)~\cite{qiu2019computing}, parallelizable global quasi-conformal map (PGQCM)~\cite{zhu2022parallelizable}, and disk Beltrami solver~\cite{guo2023automatic}.

Area-preserving parameterization can preserve the area measure of surfaces while inducing angle distortion. Compared with conformal and quasiconformal parameterizations, there are only a few works focusing on area-preserving mappings. Existing area-preserving parameterization methods include the locally authalic map~\cite{desbrun2002intrinsic}, Lie advection~\cite{zou2011authalic}, optimal mass transport (OMT)~\cite{zhao2013area,cui2019spherical}, density-equalizing map (DEM)~\cite{choi2018density,choi2020area} and stretch energy minimization (SEM)~\cite{yueh2019novel,yueh2023theoretical}.

It is worth noting that there are also various parameterization approaches aiming to balance the trade-offs between conformal and area distortions, such as the most isometric parameterizations (MIPS)~\cite{hormann2000mips}, isometric metric~\cite{smith2015bijective}, optimized conformal parameterization with controllable area distortions~\cite{lam2017optimized}, isometric distortion energy minimization~\cite{rabinovich2017scalable}, and combined quasiconformal and optimal mass transportation spherical parameterization~\cite{lyu2023two}. Besides, there are mapping methods aiming to achieve low geometric distortion while satisfying prescribed landmark constraints, such as the landmark-constrained optimized conformal parameterization~\cite{wang2005optimization,lui2007landmark, wang2007brain,lui2010optimized,choi2015flash} and quasi-conformal landmark registration (QCLR)~\cite{lam2014landmark}. Other related parameterization methods include the optimal conformal parameterization for surfaces with arbitrary topology~\cite{li2008globally}, composite majorization~\cite{shtengel2017geometric}, and the simplicial complex augmentation framework~\cite{jiang2017simplicial}.

\section{Mathematical background}\label{sec:math}
This work aims to propose a method for computing bijective density-equalizing quasiconformal maps for connected open surfaces. The method relies on the physical principle of density diffusion and quasiconformal theory. In this section, we review some basic concepts of quasiconformal theory and density-equalizing maps.

\subsection{Quasiconformal theory}

 \begin{figure}[t!]
	\centering
	\includegraphics[width=0.8\textwidth]{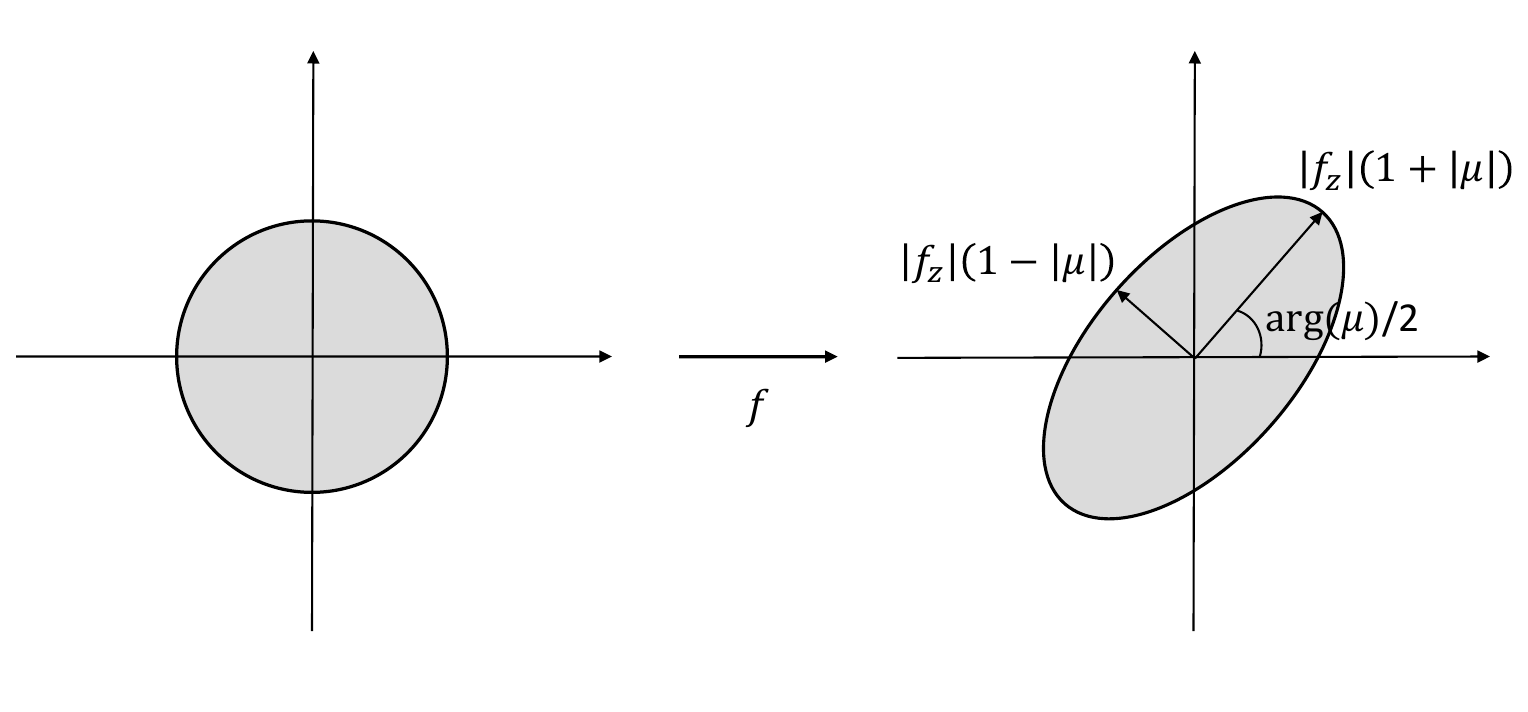}
	\caption{An illustration of how the Beltrami coefficient $\mu$ determines the quasiconformal distortion of a map $f$. Specifically, the maximal magnification factor is given by $|f_z|(1+|\mu|)$, the maximal shrinkage factor is given by $|f_z|(1-|\mu|)$, and the orientation change is given by $\text{arg}(\mu)/2$.}
	\label{fig:QCmapping} 
 \end{figure}
 
A surface $\mathcal{S}$ with a conformal structure is called a \textit{Riemann surface}. Given two Riemann surfaces $\mathcal{M}$ and $\mathcal{N}$, a map $f: \mathcal{M} \rightarrow \mathcal{N}$ is \textit{conformal} if it preserves the surface metric up to the \textit{conformal factor}. Quasiconformal mapping is an extension of conformal mapping in that it allows bounded conformal distortions. Intuitively, conformal mappings map infinitesimal circles to infinitesimal circles, while quasiconformal mappings map infinitesimal circles to infinitesimal ellipses with bounded eccentricity. Mathematically, $f: \mathbb{C} \rightarrow \mathbb{C}$ is \textit{quasiconformal} if and only if it satisfies the Beltrami equation:
\begin{equation}
\frac{\partial f}{\partial \bar{z}}=\mu(z) \frac{\partial f}{\partial z}
\end{equation}
for some complex-valued function $\mu$ satisfying $ \|\mu\|_{\infty}<1 $. $\mu$ is called the \textit{Beltrami coefficient}, which is a measure of quasiconformality, i.e. how far the map at each point deviates from a conformal map. Infinitesimally, around a point $p$, with respect to its local parameter, $f$ can be expressed as follows:
\begin{equation}
\begin{aligned}
f(z) & =f(p)+f_z(p) z+f_{\bar{z}}(p) \bar{z} =f(p)+f_z(p)(z+\mu(p) \bar{z}).
\end{aligned}
\end{equation}
It is easy to see that $f$ is conformal around $p$ if and only if $\mu(p) = 0$. Inside the local parameter domain, $f$ may be considered as a map composed of a translation to $f(p)$ together with a stretch map $S(z) = z + \mu(p) \bar z$, which is post-composed by a multiplication of $f_{z}(p)$, which is conformal. All the quasiconformal distortion of $S(z)$ is caused by $\mu(p)$. Specifically, $S(z)$ maps a small circle to a small ellipse (see Fig.~\ref{fig:QCmapping} for an illustration). Meanwhile, we can determine the angles of the directions of maximal magnification and shrinkage and their magnitude by $\mu(p)$. The maximal dilation of $f$ is given by: 
\begin{equation}
K(f)=\frac{1+\|\mu\|_{\infty}}{1-\|\mu\|_{\infty}} .
\end{equation}

Suppose $f:\Omega_1 \rightarrow \Omega_2$ and $g:\Omega_2 \rightarrow \Omega_3$ are two quasiconformal maps, with the Beltrami coefficients $\mu_f$ and $\mu_g$ respectively. Then, the Beltrami coefficient of the composition map $g\circ f : \Omega_1 \rightarrow \Omega_3$ is given by:
\begin{equation}
\mu_{g\circ f} = \frac{\mu_f + (\Bar{f_z}/f_z)(\mu_g \circ f)}{1 + (\Bar{f_z}/f_z) \Bar{\mu}_f(\mu_g \circ f)}.
\end{equation}
Angle distortions of a map can be eliminated using a suitable composition map. Specifically, for any given map $f$, by finding another map $g$ with $\mu_g = \mu_{f^{-1}}$, we can ensure that $\mu_{g\circ f} = 0$. This implies that the composition map $g\circ f$ is conformal.

For quasiconformal maps of Riemann surfaces, one can generalize the concept of Beltrami coefficients to Beltrami differentials via the local charts of the surfaces. A \textit{Beltrami differential} $\mu(z) \frac{d \bar z}{dz}$ on a Riemann surface $\mathcal{S}$ is an assignment to each chart $(U_{\alpha}, \phi_{\alpha})$ of an $L_{\infty}$ complex-valued function $\mu_{\alpha}$, defined on local parameter $z_{\alpha}$ such that 
\begin{equation}
\mu(z_{\alpha}) \frac{d \bar z_{\alpha}}{dz_{\alpha}} = \mu(z_{\beta}) \frac{d \bar z_{\beta}}{dz_{\beta}}
\end{equation}
on the domain which is also covered by another chart$(U_{\beta}, \phi_{\beta})$, where $\frac{d \bar z_{\beta}}{dz_{\beta}} = \frac{d}{d z_{\alpha}} \phi_{\alpha \beta}$ and $\phi_{\alpha \beta} = \phi_{\beta} \circ \phi_{\alpha}^{-1}$. Note that the Beltrami coefficient can be used to represent the Beltrami differential of the surface if the surface can be covered by a simple chart. As our work focuses on connected open surfaces, the Beltrami coefficient is used. 

\subsection{Linear Beltrami Solver}
Here we briefly describe the \textit{Linear Beltrami Solver (LBS)} method~\cite{lui2013texture} for reconstructing a quasiconformal map $f = u + iv$ from its Beltrami coefficient $\mu = \rho + i\tau$. Consider the real and imaginary parts of the Beltrami equation separately:
\begin{equation}
\mu(x, y)=\frac{\left(u_x-v_y\right)+i\left(v_x+u_y\right)}{\left(u_x+v_y\right)+i\left(v_x-u_y\right)}.
\end{equation}
From this formula, $v_x$ and $v_y$ can be expressed as a combination of $u_x$ and $u_y$:
\begin{equation}
\left(\begin{array}{c}
 v_y \\
-v_x
\end{array}\right)
=
A\left(\begin{array}{l}
u_x \\
u_y
\end{array}\right),
\end{equation}
where 
\begin{equation}
A=\left(\begin{array}{cc}
\alpha_1 & \alpha_2 \\
\alpha_2 & \alpha_3
\end{array}\right), \ \ \alpha_1=\frac{(\rho-1)^2+\tau^2}{1-\rho^2-\tau^2}  \ \ \alpha_2=-\frac{2 \tau}{1-\rho^2-\tau^2},  \ \ \alpha_3=\frac{1+2 \rho+\rho^2+\tau^2}{1-\rho^2-\tau^2}.
\end{equation}
Similarly, $u_x$ and $u_y$ can be expressed as
\begin{equation}
\left(\begin{array}{c}
-u_y \\
v_x
\end{array}\right)
=
A\left(\begin{array}{l}
v_x \\
v_y
\end{array}\right).
\end{equation}
Since 
\begin{equation}
\nabla \cdot\left(\begin{array}{c}
v_y \\
-v_x
\end{array}\right)= v_{xy}-v_{xy} =0 \ \ \text{ and } \ \  
\nabla \cdot\left(\begin{array}{c}
-u_y \\
u_x
\end{array}\right)= u_{xy}-u_{xy} =0,
\end{equation}
from the above equations, we have
\begin{equation}
\nabla \cdot\left(A\left(\begin{array}{l}
u_x \\
u_y
\end{array}\right)\right)=0 \ \ \text{ and } \ \  \nabla \cdot\left(A\left(\begin{array}{l}
v_x \\
v_y
\end{array}\right)\right)=0.
\end{equation}

In the discrete case, the above equations are two sparse symmetric positive definite linear systems. We can obtain $u_x$, $u_y$, $v_x$, and $v_y$ for any given Beltrami coefficient $\mu$. In this paper, we denote the above method by $f = \textbf{LBS}(\mu)$. 

\subsection{Density-equalizing maps}

\begin{figure}[t]
	\centering
	\includegraphics[width=0.8\textwidth]{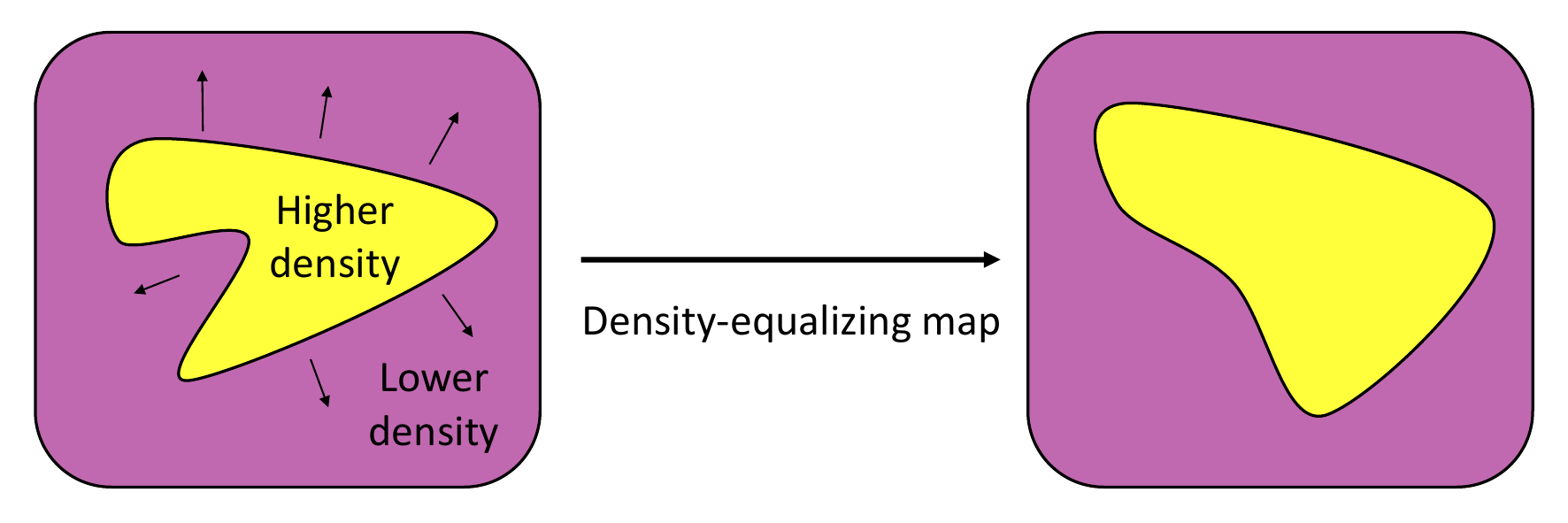}
	\caption{An illustration of the density-equalizing maps. The density diffusion creates a velocity field that enlarges regions with higher density and shrinks regions with lower density until the density is equalized.}
	\label{fig:Illustration} 
\end{figure}

Gastner and Newman~\cite{gastner2004diffusion} proposed a method for producing area cartograms based on density diffusion. Given a planar domain, a positive quantity called \textit{population} is defined at every point on it. A density field $\rho$ is then defined by the population per unit area. For instance, one can consider a geographical map of a country and define the quantity as the population of different regions. The density $\rho$ is then the population density of different regions. Density-equalizing maps aim to deform the domain by enlarging the regions with higher $\rho$ and shrinking the regions with lower $\rho$ (see Fig.~\ref{fig:Illustration} for an illustration). This can be achieved by equalizing $\rho$ through the advection equation
\begin{equation}\label{diffusion_equation}
    \frac{\partial \rho}{\partial t} +  \nabla \cdot \mathbf{j} = 0,
\end{equation}
where $\mathbf{j}$ is the density flux. By Fick's first law, the flux is given by
\begin{equation}
\mathbf{j} = - \nabla \rho.
\end{equation}
This yields the diffusion equation
\begin{equation}
\frac{\partial \rho}{\partial t} -  \Delta \rho = 0.
\end{equation}
By the definition of flux, we have $\mathbf{j} = \mathbf{v}(\mathbf{r},t) \rho$, where $\mathbf{v}(\mathbf{r},t)$ is the velocity of tracers carried by the flux. Therefore, we have
\begin{equation}
\mathbf{v}(\mathbf{r},t) = \frac{\mathbf{j}}{\rho} = - \frac{\nabla \rho}{\rho}.
\end{equation}
One can see that $\frac{\nabla \rho}{\rho}$ is independent of the absolute scale of $\rho$.

As we solve the diffusion equation to reach a steady state, the domain will be deformed according to the velocity field. Finally, the density of the domain will be equal everywhere. To describe the deformation of the domain, the displacement of the tracers is given by 
\begin{equation}
\mathbf{r}(t) = \mathbf{r}(0) + \int_{0}^{t}\mathbf{v}(\mathbf{r},\tau)d\tau,
\end{equation}
where $\mathbf{r}(0)$ is the initial state of the domain. When $t \rightarrow \infty$, the tracers $\mathbf{r}(t)$ produce the final 2D domain, which is density-equalizing. However, note that the domain will expand infinitely if there is no constraint on its boundary shape. To avoid infinite expansion, a large auxiliary region (also known as a \textit{sea}) surrounding the domain of interest is used. More specifically, by defining the density in the sea as the average density of the domain of interest and including the sea in the computation of the diffusion cartogram, the domain of interest will not expand infinitely. In~\cite{choi2018density}, Choi and Rycroft extended the idea of diffusion cartogram for surface mappings and proposed a finite-element approach for computing density-equalizing flattening maps for simply-connected open surfaces.

While the above approaches have been successfully applied to a wide range of problems, they are limited to 2D planar domains and simply-connected open surfaces but not general multiply-connected surfaces. Also, there is no theoretical guarantee for the bijectivity of the resulting density-equalizing maps in the computation. More specifically, there may be mesh fold-overs in the mapping results. Moreover, as the deformation by density-equalizing maps is merely based on the prescribed density, these methods do not preserve other local geometric structures of the surfaces. Our work aims to address these issues and achieve bijective density-equalizing quasiconformal maps.

\section{Proposed mathematical models} \label{sec: Proposed math module}

Let $\mathcal{S}$ be a multiply-connected open surface in $\mathbb{R}^3$, with $\partial \mathcal{S} = \gamma_0 -\gamma_1 - \dots -\gamma_k$ representing its boundaries, where $k \geq 1$, $\gamma_0$ is the outer boundary, and $\gamma_1, \dots, \gamma_k$ are the inner boundaries. Suppose a population is prescribed everywhere on $\mathcal{S}$. Our goal is to obtain an optimal quasiconformal map $f: \mathcal{S} \rightarrow \mathcal{D} \subset \mathbb{R}^2$ such that the final density (population per unit area) after applying $f$ becomes a constant. Here, $\mathcal{D}$ is a planar circular domain with unit outer radius, i.e. a unit disk with $k$ inner circular holes:
\begin{equation}\label{eqt:D}
    \mathcal{D} = \mathbb{D} \setminus \bigcup^{k}_{i=1} B_{r_i}(c_i),
\end{equation}
where $\mathbb D$ denotes the unit disk, $r_i \in \mathbb{R}$, $c_i \in \mathbb{C}$,  $B_{r_i}(c_i) = \{ z \in \mathbb D : |z-c_i|<r_i \}$ for $i = 1, \dots, k$, $\cup^{k}_{i = 1}B_{r_i}(c_i) \subset \mathbb D$, and
$B_{r_i}(c_i) \cap B_{r_j}(c_j)  = \emptyset$ for any $1 \leq i, j \leq k$ with $i \neq j$.

\subsection{Bijective density-equalizing quasiconformal (DEQ) map}

For ease of computation, we divide the desired optimal quasiconformal map $f: \mathcal{S} \rightarrow \mathcal{D}$ into two parts $f = \widetilde{f} \circ g_0$, where $g_0:\mathcal{S} \rightarrow \mathcal{D}^0 $ is a conformal flattening map of $\mathcal{S}$ onto a planar circular domain $\mathcal{D}^0$, and $\widetilde{f}: \mathcal{D}^0 \rightarrow \mathcal{D}$ is a planar quasiconformal map from $\mathcal{D}^0$ to an optimal planar circular domain $\mathcal{D}$. Then, the problem of finding $f$ can be simplified as the problem of finding $\widetilde{f}$ on the plane. Here, we propose a variational model for finding the optimal $\widetilde{f}$ that minimizes an energy functional $E_{\text{DEQ}}$ consisting of a density term and two Beltrami coefficient terms:
\begin{equation}
\label{eqt:DEQ}
    \widetilde{f}=\underset{g: \mathcal{D}^0 \rightarrow \mathcal{D}}{\arg \min } E_{\text{DEQ}}(g):=\underset{g: \mathcal{D}^0 \rightarrow \mathcal{D}}{\arg \min }\left\{\int_{\mathcal{D}^0} |\nabla \rho(g) |^2 + \alpha \int_{\mathcal{D}^0}| \mu(g)|^2+\beta \int_{\mathcal{D}^0}|\nabla \mu(g)|^2\right\}
\end{equation}
subject to 
\begin{align}
    & \left\|\mu(\widetilde{f})\right\|_{\infty}<1, \\
    & \left(\widetilde{f}\circ g_0\right) (\gamma_i) = \partial B_{r_i}(c_i) \quad \text{for} \quad 0 \leq i \leq k.
\end{align}
In this context, $\nabla \rho$ denotes the gradient of the density function $\rho$, $\mu(g)$ and $\mu(\widetilde{f})$ represent the Beltrami coefficients of $g$ and $\widetilde{f}$, respectively, and $\partial B_{r_0}(c_0)$ denotes the outer boundary. Here, $\alpha$ and $\beta$ are two user-defined nonnegative weighting parameters for controlling the quasiconformality. Specifically, the first term $\int_{\mathcal{D}^0} |\nabla \rho(g) |^2$ ensures that the spatial variation of the density function $\rho$ decreases. The second term $\int_{\mathcal{D}^0} |\mu(g)|^2$ helps reduce the quasiconformal distortion and ensure the bijectivity of the mapping. Meanwhile, the third term $\int_{\mathcal{D}^0} |\nabla \mu(g)|^2$ helps enhance the smoothness of the mapping. As our focus is on the density-equalizing effect of the mapping, both $\alpha$ and $\beta$ should be small. We call the resulting composition map $f = \widetilde{f} \circ g_0$ the \emph{bijective density-equalizing quasiconformal (DEQ) map} for $\mathcal{S}$.

The Jacobian $J_{\widetilde{f}}$ of $\widetilde{f}$ can be represented by
\begin{equation}
J_{\widetilde{f}} = \left|\frac{\partial \widetilde{f}}{\partial z}\right|^2(1-|\mu(\widetilde{f})|^2).
\end{equation}
One can see that the constraint on the norm of $\mu$ implies $J_{\widetilde{f}}>0$ everywhere. By the inverse function theorem, the mapping $\widetilde{f}$ is thus locally diffeomorphic everywhere. Meanwhile, $|J_{\widetilde{f}}|>0$ also implies that $\widetilde{f}$ induces a Riemannian metric on $\mathcal{S}$. More specifically, an auxiliary metric $|dz+\mu \overline{dz}|^2$ can be induced by the Beltrami coefficient $\mu$ on $\mathcal{D}^0$. Note that $\widetilde{f}$ is a conformal map under the auxiliary metric. By the Riemann mapping theorem for multiply-connected domains~\cite{henrici1993applied}, every multiply-connected Riemann surface with a valid Riemannian metric can be mapped by a conformal and one-to-one transformation onto a planar circular domain. Therefore, $\widetilde{f}$ is a global bijection on $\mathcal{D}^0$. It is easy to see that if $\widetilde{f}$ is locally diffeomorphic everywhere and globally bijective, $\widetilde{f}$ is a diffeomorphism. Consequently, the composition map $f = \widetilde{f} \circ g_0$ is also a diffeomorphism.

Moreover, a minimizer of the above optimization problem \eqref{eqt:DEQ} is guaranteed to exist by the following proposition.

\begin{proposition}\label{proposition1}
(Existence of the minimizer of $E_{\text{DEQ}}$)
    Let
    \begin{equation}
    \begin{aligned}
        \mathcal{A} = \{\nu \in C^1(\mathcal{D}^0): & \|D \nu\|_{\infty} \leq C_1 ; \ \|\nu\|_{\infty} \leq 1-\epsilon ; \  \nu \text{ is the Beltrami coefficient} \\ & \text{of a quasiconformal map} \}
    \end{aligned}
    \end{equation}
    for some $C_1>0$ and small $\epsilon > 0$, where $\mathcal{D}^0$ is a circular domain. Then $E_{\text{DEQ}}$ has a minimizer in $\mathcal{A}$. In fact, $\mathcal{A}$ is compact.
\end{proposition}
    
\begin{proof}
    Here we follow the approach in~\cite{lam2014landmark} to prove that $\mathcal{A}$ is complete and totally bounded. First, note that the existence of conformal maps (with Beltrami coefficient $\nu = 0$) implies that $\mathcal{A}$ is non-empty. Now, we prove that $\mathcal{A}$ is complete. Define $\|\mu\|_s = \|\mu\|_{\infty} + \|D \mu\|_{\infty}$. Let $\{\nu_n\}_{n=1}^{\infty}$ be a Cauchy sequence in $\mathcal{A}$ under the norm $\|\cdot\|_s$. 
    It is worth noting that $\{\nu_n\}_{n=1}^{\infty}$ and $\{D\nu_n\}_{n=1}^{\infty}$ are two Cauchy sequences with respect to the $\|\cdot \|_{\infty}$ norm. Since $\mathcal{A} \subset W^{1,\infty}\left( \mathcal{D}^0 \right)$ and $ W^{1,\infty}$ is complete, $D\nu_n$ converges uniformly to some $h \in W^{1,\infty}\left( \mathcal{D}^0 \right)$. Moreover, $h$ is continuous since every $D\nu_n$ is continuous. Besides, $\nu_n$ is uniformly convergent in $ W^{1,\infty}$ as well. Based on the above, $\nu_n \rightarrow \nu$ and $D \nu_n \rightarrow D\nu$ uniformly for some $\nu \in C^1\left( \mathcal{D}^0 \right)$. Additionally, $\|\nu_n \|_{\infty} \leq 1-\epsilon$ and $\| D \nu_n \|_{\infty} \leq C_1$ for every $n$ imply that $\| \nu \|_{\infty} \leq 1-\epsilon$ and $\| D \nu \|_{\infty} \leq C_1$, respectively. Therefore, $\nu \in \mathcal{A}$. Consequently, $\mathcal{A}$ is complete.
    
    Next, we prove that $\mathcal{A}$ is totally bounded. To achieve this, we will show that for any given $\eta > 0$, there exists a finite collection of open balls with radius $\eta$. These balls will have their centers located within $\mathcal{A}$, and their union will cover the entire $\mathcal{A}$. First, we construct a net for $\mathcal{A}$. Note that $\mathcal{D}^0$ represents a planar circular domain with circular holes. Let $\Bar{\Omega}$ be the unit rectangle and $\mathcal{D}^0 \subset \Bar{\Omega}$. Now, we construct a regular grid on $\Bar{\Omega}$ with edge length $\frac{1}{m}$. After re-indexing, we denote the grid points as $\{p_i\}_{i \in I}$. Next, we define a smooth tent function on each grid point as follows: $ \mathcal{T}_{m,n,L} = \{ \left.\frac{k}{n}\Tilde{\delta}^{L}_{p_i} \right|_{\mathcal{D}^0}  \}_{i \in I,0 \leq k \leq n} $. In this context, $\Tilde{\delta}^{L}_{p_i}$ represents the smooth approximation of the tent function $\delta_{p_i}$ subject to the conditions $\| D \Tilde{\delta}^L_{p_i} - m\mathbf{Id} \| _{\infty} + \| \Tilde{\delta}^L_{p_i} - \delta_{p_i} \| _{\infty} \leq \frac{1}{L}$ and $\|\Tilde{\delta}^L_{p_i} \| _{\infty} \leq C_1$. Now we consider: $\mathcal{B}_{m,n,L} = \{ \tilde \nu \in \mathcal{A}: \tilde \nu = \sum_{i} T_i, \  T_i \in  \mathcal{T}_{m,n,L} \}$. It is important to note that due to the finite number of grid points, the functions in $\mathcal{B}_{m,n,L}$ are also finite in number. Finally, we construct a finite union cover of $\mathcal{A}$. For any $\eta > 0$ and any $\nu \in \mathcal{A}$, we can choose some sufficiently large numbers $m$, $n$, $L$ such that there exists a sum of tent functions $ \Bar{\nu} = \sum_{i}  \left.\frac{k}{n}\delta_{p_i} \right|_{\mathcal{D}^0}$ and $\tilde \nu \in \mathcal{B}_{m,n,L}$ satisfying $\| \nu - \Bar{\nu} \|_{\infty} < \frac{\eta}{4}$, $\| \Bar{\nu} - \tilde \nu \|_{\infty} < \frac{\eta}{4}$, $\| D\nu - D\Bar{\nu} \|_{\infty} < \frac{\eta}{4}$, and $\| D\Bar{\nu} - D\tilde \nu \|_{\infty} < \frac{\eta}{4}$. Then we have 
    $\| \nu - \tilde \nu \|_{s} \leq \| \nu - \Bar{\nu}\|_{s} + \| \Bar{\nu} - \tilde \nu \|_{s} < \eta$. Hence, $\mathcal{A} \subset \cup_{\Tilde{\nu}\in \mathcal{B}_{m,n,L}} B_{\eta}(\Tilde{\nu})$. Here, $B_{\eta}(\Tilde{\nu})$ is the open neighborhood ball of $\Tilde{\nu}$ with radius $\eta$. Consequently, $\mathcal{A}$ is a totally bounded set.

    Since $\mathcal{A}$ is complete and totally bounded, $\mathcal{A}$ is compact. Also, note that $\nabla \rho $ is a $L^2$ function and $E_{\text{DEQ}}$ is continuous in $\mathcal{A}$. Hence, $E_{\text{DEQ}}$ has a minimizer in $\mathcal{A}$.
\end{proof}

With the theoretical existence of the minimizer of \eqref{eqt:DEQ} proved, in Section~\ref{sect:DEQ} we develop an iterative algorithm to solve for the optimal planar circular domain $\mathcal{D}$ and the optimal density-equalizing quasiconformal map $f:\mathcal{S}\to \mathcal{D}$.

\subsection{Bijective landmark-matching density-equalizing quasiconformal (LDEQ) map}
In many practical situations, it is desirable to compute flattening maps with certain feature points on a surface mapped to some prescribed positions. We can achieve this by incorporating landmark-matching constraints into the bijective density-equalizing quasiconformal model.

Let $\{p_i\}^m_{i = 1}$ be a set of feature landmarks defined on $\mathcal{S}$ and $\{ q_i\}^m_{i = 1}$ be the corresponding target positions defined on the plane. Our objective is to find an optimal planar domain $\mathcal{D}$ and a \emph{bijective landmark-matching density-equalizing quasiconformal (LDEQ)} map $f: \mathcal{S} \rightarrow \mathcal{D}$ satisfying $f(p_i) = q_i$ for all $i = 1,\dots,m$. 

To achieve this, we again consider $f = \widetilde{f} \circ g_0$ where $g_0: \mathcal{S} \to \mathcal{D}^0$ is a conformal flattening map of $\mathcal{S}$ onto a planar circular domain $\mathcal{D}^0$ and $\widetilde{f}: \mathcal{D}^0 \to \mathcal{D}$ is a planar quasiconformal map. We propose a variational approach to obtain $\widetilde{f}$ by minimizing the following energy functional $E_{\text{LDEQ}}$:
\begin{equation}
\label{eqt:LDEQ}
    \widetilde{f}=\underset{g: \mathcal{D}^0 \rightarrow \mathcal{D}}{\arg \min } E_{\text{LDEQ}}(g):=\underset{g: \mathcal{D}^0 \rightarrow \mathcal{D}}{\arg \min }\left\{\int_{\mathcal{D}^0} |\nabla \rho(g) |^2 + \alpha \int_{\mathcal{D}^0}| \mu(g)|^2+\beta \int_{\mathcal{D}^0}|\nabla \mu(g)|^2\right\}
\end{equation}
subject to:
\begin{align}
    & \left(\widetilde{f}\circ g_0\right)(p_i) = q_i \quad \text{for} \quad 1 \leq i \leq m, \\
    & \left(\widetilde{f}\circ g_0\right)(\gamma_i) = \partial B_{r_i}(c_i) \quad \text{for} \quad 0 \leq i \leq k, \\
    & \|\mu(\widetilde{f})\|_{\infty}<1.
\end{align}
Again, here $\mu(\widetilde{f})$ and $\mu(g)$ denote the Beltrami coefficient of $\widetilde{f}$ and $g$, respectively, $\partial B_{r_0}(c_0)$ denotes the outer boundary, and $\alpha$ and $\beta$ are two user-defined small nonnegative weighting parameters for controlling the quasiconformality.

Analogous to the previous model, the existence of the solution to the optimization problem \eqref{eqt:LDEQ} can be proved as follows.
\begin{proposition}\label{proposition2}
(Existence of the minimizer of $E_{\text{LDEQ}}$)
    Let
    \begin{equation}
    \begin{aligned}
        \mathcal{A} = \{ \nu \in C^1(\mathcal{D}^0): & \|D \nu\|_{\infty} \leq C_1 ;  \ \|\nu\|_{\infty} \leq 1-\epsilon ;  \ \nu \text{ is the Beltrami coefficient of a }\\& \text{quasiconformal map} \enspace g^{\nu} \enspace \text{with} \enspace g^{\nu}(z_i) = w_i \text{ for all }  1 \leq i \leq m \}
        \end{aligned}
    \end{equation}
    for some $C_1>0$ and small $\epsilon > 0$, where $\mathcal{D}^0$ is a circular domain, $\{z_i\}_{i=1}^m$ is a set of points on $\mathcal{D}^0$, and $\{w_i\}_{i=1}^m$ is a set of points on $\mathcal{D}$. Then $E_{\text{LDEQ}}$ has a minimizer in $\mathcal{A}$. In fact, $\mathcal{A}$ is compact.
\end{proposition}

\begin{proof}
    Note that $\mathcal{A} \neq \emptyset$. Now, we prove that $\mathcal{A}$ is complete. Define $\|\mu\|_s = \|\mu\|_{\infty} + \|D \mu\|_{\infty}$. Let $\{\nu_n\}_{n=1}^{\infty}$ be a Cauchy sequence in $\mathcal{A}$ under the norm $\|\cdot\|_s$. Then, $\{D\nu_n\}_{n=1}^{\infty}$ is also a Cauchy sequence with respect to the $\|\cdot \|_{\infty}$ norm. Since $\mathcal{A} \subset W^{1,\infty}(\mathcal{D}^0)$ and $W^{1,\infty}$ is complete, $D\nu_n \rightarrow h$ uniformly for some $h\in W^{1,\infty}(\mathcal{D}^0)$. Since $D\nu_n$ is continuous, $h$ is also continuous. Moreover, $\nu_n$ is convergent in $ W^{1,\infty}(\mathcal{D}^0)$ as well. Hence, $\nu_n \rightarrow \nu$ and $D \nu_n \rightarrow D\nu$ uniformly for some $\nu \in C^1(\mathcal{D}^0)$. In addition, $\|\nu_n \|_{\infty} \leq 1-\epsilon$ for all $n$ implies that $\| \nu \|_{\infty} \leq 1-\epsilon$. Also, $\| D \nu_n \|_{\infty} \leq C_1$ implies that $\| D \nu \|_{\infty} \leq C_1$. Since $\nu_n \rightarrow \nu$ uniformly, $g^{\nu_n} \rightarrow g^{\nu}$ locally uniformly. More specifically, we have $g^{\nu}(z_i) = w_i$ for $i = 1,2,\dots ,m$. Therefore, $\nu \in \mathcal{A}$. Consequently, $\mathcal{A}$ is Cauchy complete.

    Next, we prove that $\mathcal{A}$ is totally bounded. It is easy to see that $\mathcal{A}$ is a subset of $\widetilde{\mathcal{A}} = \{\nu \in C^1(\mathcal{D}^0): \|D \nu\|_{\infty} \leq C_1 ; \ \|\nu\|_{\infty} \leq 1-\epsilon; \  \nu$ is the Beltrami coefficient of a quasiconformal map$\}$, and we have already proved that $\widetilde{\mathcal{A}}$ is totally bounded in the proof of Proposition~\ref{proposition1}. Since any subset of a totally bounded set is totally bounded, $\mathcal{A}$ is totally bounded.

    As a result, $\mathcal{A}$ is compact. Therefore, $E_{\text{LDEQ}}$ has a minimizer in $\mathcal{A}$.
\end{proof}

In Section~\ref{sect:LDEQ}, we develop an iterative algorithm for finding the optimal domain $\mathcal{D}$ and the landmark-matching density-equalizing quasiconformal map $f:\mathcal{S}\to \mathcal{D}$. 

\section{Proposed algorithms}\label{sec:main}
Consider a given multiply-connected open surface $\mathcal{S}$ in $\mathbb R^3$ with $\partial S = \gamma_0-\gamma_1 - \dots - \gamma_k$ discretized as a triangle mesh $(\mathcal{V}, \mathcal{E}, \mathcal{F})$, where $\mathcal{V}$ is the set of vertices, $\mathcal{E}$ is the set of edges, and $\mathcal{F}$ is the set of triangular faces. Suppose a population is assigned to each triangle face of the mesh. We first apply an initial conformal flattening map to map $\mathcal{S}$ onto a planar circular domain. The face density can then be defined as $\rho^{\mathcal{F}} = \frac{\text{Population}}{\text{Area of the triangle}}$ on the circular domain. One can also define the point density $\rho^{\mathcal{V}}$ by considering the 1-ring face neighborhood of each vertex (see Appendix~\ref{appendix:density} for the detailed formulation).

Below, we develop two algorithms based on the models~\eqref{eqt:DEQ} and~\eqref{eqt:LDEQ} for computing bijective density-equalizing quasiconformal (DEQ) maps and bijective landmark-matching density-equalizing quasiconformal (LDEQ) maps.

\subsection{Bijective density-equalizing quasiconformal map}\label{sect:DEQ}

Recall that our objective is to find a circular domain $\mathcal{D} \in \mathbb R^2$ and a diffeomorphism $f:\mathcal{S} \rightarrow \mathcal{D}$ that satisfies two criteria: first, $f$ should be a bijective density-equalizing map between $\mathcal{S}$ and $\mathcal{D}$, and second, it should minimize the local geometric distortion. Equivalently, we want to find an optimal Beltrami coefficient such that the diffeomorphism associated with it satisfies the two criteria above. Due to the presence of higher-order derivative terms, it is challenging to directly solve the optimization problem~\eqref{eqt:DEQ} to get the mapping. Instead, we consider using the Beltrami coefficient as the variable of the optimization problem. The variational problem \eqref{eqt:DEQ} for $\widetilde{f}$ can be transformed into the problem of computing the optimal Beltrami coefficient $\nu^{*}$ that minimizes the following energy functional $E_{\text{DEQ}}$:
\begin{equation}\label{eqt:DEQ_nu}
    E_{\text{DEQ}}(\nu):=\int_{\mathcal{D}^0} |\nabla \rho(g) |^2 + \alpha \int_{\mathcal{D}^0}| \nu|^2+\beta \int_{\mathcal{D}^0}|\nabla \nu|^2
\end{equation}
subject to  
\begin{align}
&\|\nu\|_{\infty}  <1, \\
     &\nu = \mu(g),
\end{align}
where $\mathcal{D}^0$ is the initial flatten circular domain, $g$ is a quasiconformal map, $\mu(g)$ is the Beltrami coefficient of $g$, and $\rho$ is the density function given by the population defined on $\mathcal{S}$.  

Below, we develop an algorithm for finding the optimal map based on the above formulation. Our proposed algorithm involves three major components:
\begin{itemize}
    \item \textbf{Initial conformal flattening map}: We first conformally flatten $\mathcal{S}$ onto a planar circular domain $\mathcal{D}^0$ with $k$ inner circular holes (Section~\ref{sect:initial}).
    \item \textbf{Geometry modification}: We compute a modification mapping and obtain a new circular domain with the size and position of the inner boundaries adjusted (Section~\ref{sect:GM}).
    \item \textbf{Beltrami density-equalizing descent}: We search for a density-equalizing descent direction for the Beltrami coefficient on the new circular domain and reconstruct an updated mapping using it  (Section~\ref{sect:BDED}).
\end{itemize}

\subsubsection{Initial conformal flattening map}\label{sect:initial}
To begin, we compute an initial flattening map $g_0:\mathcal{S} \to \mathcal{D}^0$, where $\mathcal{D}^0 = \mathbb{D} \setminus \cup^{k}_{i=1} B_{r_i}(c_i)$ is a planar circular domain with $k$ inner circular holes as described in Eq.~\eqref{eqt:D}, where $r_i$ and $c_i$ are the radius and center of the $i$-th inner circular hole of $\mathcal{D}^0$, and $g_0(\gamma_i) = \partial B_{r_i}(c_i)$ for all $i = 1, \dots, k$.

In particular, here we look for a conformal flattening map as the initialization. Recall that our goal is to achieve a density-equalizing map with reduced quasiconformal distortion. By starting with a conformal parameterization, we can ensure that the initial quasiconformal distortion is low and the triangulations are regular. This facilitates the subsequent computation of the density-equalizing quasiconformal map. 

To achieve this, various existing parameterization methods for multiply-connected open surfaces can be applied. For instance, the Koebe's iteration method~\cite{koebe1910konforme} can be used.

\subsubsection{Geometry modification of the circular domain}\label{sect:GM}
After getting the initial flattening map, we can then look for a density-equalizing map on the plane. However, it is worth noting that while all inner holes of the flattened domain $\mathcal{D}^0$ are perfectly circular, their size and position may not be optimal for the density diffusion process. To improve the accuracy and efficiency of the density-equalizing map, we consider modifying the geometry of $\mathcal{D}^0$ by changing the centers and radii of its inner circular holes. More specifically, our geometry modification approach involves two processes: \emph{domain alteration} and \emph{domain reconstruction} (see Fig.~\ref{fig:illustration_gm}). After the modification, we get a new circular domain $\mathcal{D}^1$ with the position and size of the circular holes updated. 

\begin{figure}[t]
    \centering
    \includegraphics[width=\textwidth]{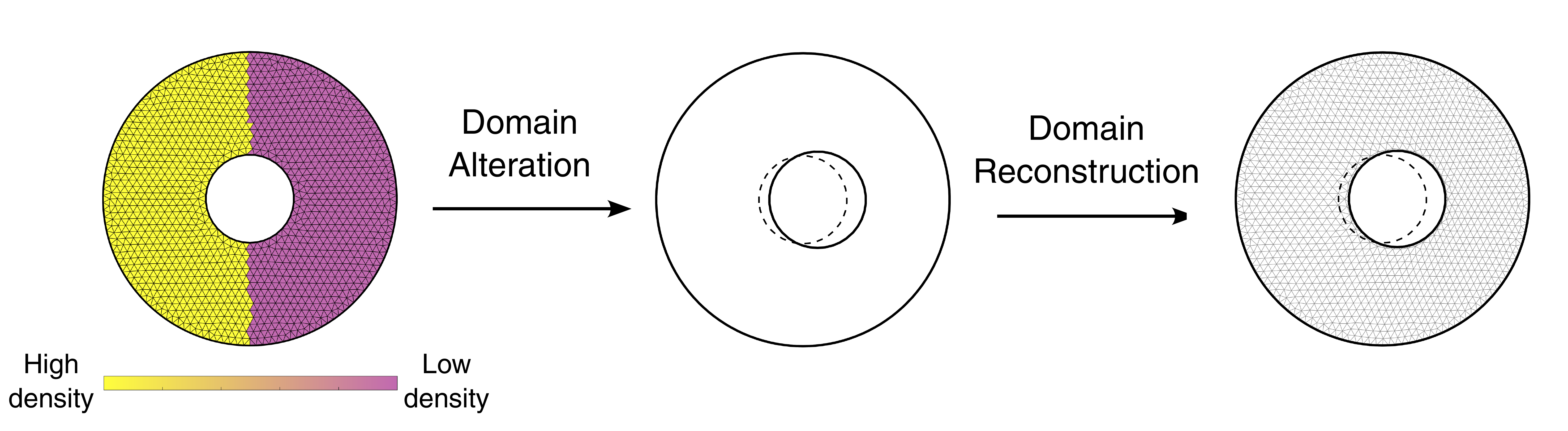}
    \caption{An illustration of the geometry modification (GM) algorithm. Given a planar circular domain, we first change the size and position of the inner circular holes based on the density gradient. Then, we reconstruct a smooth quasiconformal map based on the updated boundaries. Here, the dashed curve represents the previous inner boundary, and the solid curve represents its updated position.}
    \label{fig:illustration_gm}
\end{figure}

We first consider the domain alteration, which is performed on the boundaries of $\mathcal{D}^0$. In this step, our goal is to modify the boundaries $\partial \mathcal{D}^0$ and get a new set of circular holes $\partial \mathcal{D}^1$ such that they are more suitable for the subsequent density-equalizing process. Therefore, the holes $B_{r_i}(c_i) \ (i = 1,2,\dots,k)$ should deform naturally based on the density gradient in the domain while remaining to be circular. In other words, the deformation of each hole must be a combination of translation, scaling, and rotation, where the translation of the hole can be considered as a translation of the hole center, the scaling can be considered as a movement of all points along the radial direction, and the rotation can be considered as a movement of all points along the tangential direction. This motivates us to consider minimizing the following domain modification functional:
\begin{equation}\label{eqt:modify_function}
    E_{\text{bdy}}(\nu,\mathbf{c},\mathbf{r},\mathbf{T}):=\int_{\partial \mathcal{D}^0} |\nabla \rho(\widetilde{g},\mathbf{c},\mathbf{r},\mathbf{T}) |^2 + \alpha \int_{\partial \mathcal{D}^0}| \nu(\mathbf{c},\mathbf{r},\mathbf{T})|^2+\beta \int_{\partial \mathcal{D}^0}|\nabla \nu(\mathbf{c},\mathbf{r},\mathbf{T})|^2,
\end{equation}
subject to $\|\nu\|_{\infty}<1$ and $\nu = \mu(\widetilde{g})$. Here, $\widetilde{g}: \mathcal{D}^0 \rightarrow \mathcal{D}^1$ is the desired modification map, $\mathbf{c}$ is the set of centers of all inner boundaries of $\mathcal{D}^0$, $\mathbf{r}$ is the set of radii of all inner boundaries, and $\mathbf{T}$ is the set of tangent vectors on all inner boundaries. More specifically, $\mu(\widetilde{g})$ is the Beltrami coefficient of $\widetilde{g}$ and $\rho(\widetilde{g},\mathbf{c},\mathbf{r},\mathbf{T})$ is the density function. 
For simplicity, we denote
\begin{align}
    E^1_{\text{bdy}} &= \int_{\partial \mathcal{D}^0} |\nabla \rho(\widetilde{g},\mathbf{c},\mathbf{r},\mathbf{T}) |^2, \\ E^2_{\text{bdy}} &=  \int_{\partial \mathcal{D}^0}| \nu(\mathbf{c},\mathbf{r},\mathbf{T})|^2, \\ E^3_{\text{bdy}} &=  \int_{\partial \mathcal{D}^0}|\nabla \nu(\mathbf{c},\mathbf{r},\mathbf{T})|^2.
\end{align}

Now, we first compute the alteration of $E^1_{\text{bdy}}$. By the diffusion equation and Fick's law, the diffusion-based velocity is given by: 
\begin{equation} \label{eqt:rho_velocity}
\mathbf{v} = -\frac{\nabla \rho^{\mathcal{V}}}{\rho^{\mathcal{V}}},
\end{equation}
where $\rho^{\mathcal{V}}$ is the vertex density (see Appendix~\ref{appendix:density} for the discretization and computation).

Instead of using the gradient, we utilize the velocity $\mathbf{v}$, which points towards the global minimum, as the direction of motion. Therefore, when deforming the inner boundaries, we aim to find suitable directions for the centers and boundary points that are in proximity to the velocity vector $\mathbf{v}$. This ensures that the points move towards the global minimum. The minimization of the energy $E^1_{\text{bdy}}$ can be interpreted as the minimization of the following energy:
\begin{equation}\label{eq:move}
    \widetilde{E}^1_{\text{bdy}}(\mathbf{c},\mathbf{r},\mathbf{T}) = \sum_{i=1}^k \int_{\partial B_{r_{i}}(c_i)}  |v_i-\mathbf{v}|^2 = \sum_{i=1}^k \sum^{n_i}_{j = 1}  |v_{i,j}-\mathbf{v}_{i,j}|^2
\end{equation}
subject to the constraint that all boundaries $\{ \partial B_{r_i}(c_i)\}^k_{i = 1}$ remain circular after the deformation. Here, $n_i$ represents the number of vertices on the boundary $\partial {B_{r_i}(c_i)}$, and $v_{i,j}$ and $\mathbf{v}_{i,j}$ denote the velocity $v$ and $\mathbf{v}$ at the $j$-th vertex on the $i$-th circular hole respectively.

Now, for any point on an inner circular hole, it can be expressed as $c  + r e^{i\theta}$ for some $c,r,\theta$. Similarly, its updated position under a deformation can be expressed as $(c + \Delta c) + (r+ \Delta r) e^{i(\theta + \Delta \theta)}$ for some $\Delta c,\Delta r,\Delta \theta$. Therefore, the velocity $v_i$ at any boundary point on the $i$-th inner circular hole can be decomposed into three distinct components: a translation of the hole center, a scaling of the hole along the radial direction of the hole, and a rotation along the tangential direction of the hole. In other words, we have 
\begin{equation}
v_i = v_i^c +  v_i^r +  v_i^T,
\end{equation}
where $v_i^c$ is the translational velocity associated with the center, $v_i^r$ is the radial velocity, and $v_i^T$ is the tangential velocity. We can then rewrite Eq.~\eqref{eq:move} as follows:
\begin{equation}\label{eq:velocity}
\begin{aligned}
    &\widetilde{E}^1_{\text{bdy}}(\mathbf{c},\mathbf{r},\mathbf{T}) \\ = &\sum^k_{i = 1} \left(\int_{ \partial B_{r_i}(c_i)} |v_i^c - \mathbf{v}^c_i|^2 + \int_{\partial B_{r_i}(c_i)}|v_i^r - \mathbf{v}_i^r|^2 +  \int_{\partial B_{r_i}(c_i)}|v_i^T - \mathbf{v}_i^T|^2\right) \\
     = &\sum^k_{i = 1} \sum^{n_i}_{j = 1} \left( |v_{i,j}^c - \mathbf{v}^c_{i,j}|^2 +  |v^r_{i,j} - \mathbf{v}^r_{i,j}|^2 +  |v^T_{i,j} - \mathbf{v}^T_{i,j}|^2 \right),
\end{aligned}
\end{equation}
where ${v}^c_{i,j}$, ${v}^r_{i,j}$, ${v}^T_{i,j}$ and $\mathbf{v}^c_{i,j}$, $\mathbf{v}^r_{i,j}$, $\mathbf{v}^T_{i,j}$ denote the translational, radial and tangential velocity components of $v$ and $\mathbf{v}$ at the $j$-th vertex on the $i$-th inner circular hole respectively.

To ensure that the inner boundary $\partial B_{r_i}(c_i)$ is still a circular hole after translation, the translation velocity $v_{i,j}^c$ should be the same for all boundary points, and hence we can simply denote it as $v_i^c$ without ambiguity. The radial velocity $v^r_{i,j}$ must be equal to $p^r_i \vec r_{i,j}$ for all boundary points, where $p^r_i$ is a real number and ${\vec{r}}_{i,j}$ is the outer unit normal vector at each point. Similarly, the tangent velocity $v^T_{i,j}$ can be denoted as $p^T_i \vec T_{i,j}$, where $p^T_i$ is a real number and $\vec T_{i,j}$ is the unit tangent vector at each point. Hence, the optimization problem \eqref{eq:velocity} can be rewritten as:
\begin{equation}\label{eq:velocity1}
   \widetilde{E}^1_{\text{bdy}}(\mathbf{c},\mathbf{p}^r,\mathbf{p}^T) =   \sum^k_{i = 1} \sum_{j = 1}^{n_i} \left( \left|v^c_i -\mathbf{v}^c_{i,j} \right|^2 + \left|p^r_i \vec r_{i,j} - \mathbf{v}^r_{i,j}\right|^2 +  \left|p^T_i \vec T_{i,j} - \mathbf{v}^T_{i,j}\right|^2\right),
\end{equation}
where $\mathbf{p}^r = (p^r_1,\dots,p^r_k)$ and $\mathbf{p}^T = (p^T_1,\dots,p^T_k)$. Now, for the $i$-th inner circular hole, we denote 
\begin{equation}
    \widetilde{E}^1_{\text{bdy}_i}(c_i,p^r_i,p^T_i) = \sum_{j = 1}^{n_i} \left|v^c_i -\mathbf{v}^c_{i,j} \right|^2 + \sum_{j = 1}^{n_i} \left|p^r_i \vec r_{i,j} - \mathbf{v}^r_{i,j}\right|^2 + \sum_{j = 1}^{n_i} \left|p^T_i \vec T_{i,j} - \mathbf{v}^T_{i,j}\right|^2.
\end{equation}
Using the principle of minimum energy and Eq.~\eqref{eqt:rho_velocity}, the minimum values of the three parts in $\widetilde{E}^1_{\text{bdy}_i}(c_i,p^r_i,p^T_i)$ can be obtained by:
\begin{align}
& \widetilde{v}^c_i  = \frac{1}{n_i} \sum_{j = 1}^{n_i} \left( \mathbf{v}^c_{i,j} \right) =  \frac{1}{n_i} \sum_{j = 1}^{n_i} \left(-\frac{\nabla \rho^{\mathcal{V}}(w_{i,j})}{\rho^{\mathcal{V}}(w_{i,j})}\right)^c_{i,j},\\
& \widetilde{p}^r_i = \frac{1}{n_i} \sum_{j = 1}^{n_i} \left( \mathbf{v}^r_{i,j} \right) \cdot \vec r_{i,j} = \frac{1}{n_i} \sum_{j = 1}^{n_i} \left(-\frac{\nabla \rho^{\mathcal{V}}(w_{i,j})}{\rho^{\mathcal{V}}(w_{i,j})}\right)^r_{i,j} \cdot \vec r_{i,j},\\
& \widetilde{p}^T_i = \frac{1}{n_i} \sum_{j = 1}^{n_i} \left( \mathbf{v}^T_{i,j} \right) \cdot \vec T_{i,j} = \frac{1}{n_i} \sum_{j = 1}^{n_i} \left(-\frac{\nabla \rho^{\mathcal{V}}(w_{i,j})}{\rho^{\mathcal{V}}(w_{i,j})}\right)^T_{i,j} \cdot \vec T_{i,j},
\end{align}
where $\rho^{\mathcal{V}}(w_{i,j})$ is the vertex density at the boundary point $w_{i,j}$. The descent direction associated with the energy $E^1_{\text{bdy}}$ for all inner boundaries is given by:
\begin{equation}\label{eqt:dw1}
    dw^1 = \sum^k_{i=1} \sum_{j = 1}^{n_i} (\widetilde{v}_{c_i} +  \widetilde{p}^r_i \vec r_{i,j} +  \widetilde{p}^T_i \vec T_{i,j}). 
\end{equation}

Until now, we have only dealt with the optimization problem \eqref{eq:move}. However, as the value of $dw^1$ is solely linked to the diffusion process, the bijectivity and geometry structure are not preserved in general. Therefore, our next step is to consider the other two energy terms $E^2_{\text{bdy}}$ and $E^3_{\text{bdy}}$ on the inner boundaries.

First, we express the energy $E^2_{\text{bdy}}$ as follows:
\begin{equation}
E^2_{\text{bdy}}(\mathbf{c},\mathbf{r},\mathbf{T}) =  \sum^k_{i=1} \sum_{\mathcal{T}_l \in \mathcal{N}^{\mathcal{F}}(\gamma_i)} \operatorname{Area}(\mathcal{T}_l)  |\nu_l|^2,
\end{equation}
where $\mathcal{N}^{\mathcal{F}}(\gamma_i)$ is the 1-ring face neighborhood of the $i$-th inner boundary, and $\nu_l$ is the boundary Beltrami coefficient associated with $\mathcal{T}_l$. Specifically, $\nu_l$ can be expressed in terms of the center $c_i$, radius vector $r_{i,j}$ and tangent vector $T_{i,j}$ of some vertex on the boundary. The detailed computation of $\nu_l = \nu_l(c_i,r_{i,j},T_{i,j})$ is provided in Appendix~\ref{appendix:BC}.

The descent directions associated with $E^2_{\text{bdy}}$ for the $i$-th inner hole can be obtained by:
\begin{align}
    &d w^2_{i,1} = \frac{1}{n_i}\sum^{n_i}_{j=1} \sum_{\mathcal{T}_l \in \mathcal{N}^{\mathcal{F}}_{i,j}} 2\operatorname{Area}(\mathcal{T}_l)\left(\nu_l(c_i,r_{i,j},T_{i,j})\frac{\partial \nu_l}{\partial r_{i,j}}\right),\\
    &d w^2_{i,2} = \frac{1}{n_i}\sum^{n_i}_{j=1} \sum_{\mathcal{T}_l \in \mathcal{N}^{\mathcal{F}}_{i,j}} 2\operatorname{Area}(\mathcal{T}_l)\left(\nu_l(c_i,r_{i,j},T_{i,j})\frac{\partial \nu_l}{\partial T_{i,j}}\right),\\
    &d w^2_{i,3} = \frac{1}{n_i}\sum^{n_i}_{j=1} \sum_{\mathcal{T}_l \in \mathcal{N}^{\mathcal{F}}_{i,j}} 2\operatorname{Area}(\mathcal{T}_l)\left(\nu(c_i,r_{i,j},T_{i,j})\frac{\partial \nu_l}{\partial c_i}\right).    
\end{align}
Here, $\mathcal{N}^{\mathcal{F}}_{i,j}$ is the set of 1-ring neighborhood triangle elements that contain the $j$-th vertex of the $i$-th inner hole, the partial derivatives $\frac{\partial \nu_l}{\partial r_{i,j}}$ and $\frac{\partial \nu_l}{\partial T_{i,j}}$ correspond to the descent direction along the radial and tangential directions, respectively, and $\frac{\partial \nu_l}{\partial c_i}$ represents the descent direction associated with the center of the circular hole.

Combining the above formulas for $dw^2_{i,1}, dw^2_{i,2}, dw^2_{i,3}$ for $i = 1, \dots, k$, the final descent direction associated with the energy $E^2_{\text{bdy}}$ for all inner boundaries is given by: 
\begin{equation} \label{eqt:dw2}
dw^2 =  \sum^k_{i = 1} \left(d w^2_{i,1} +  d w^2_{i,2} +  d w^2_{i,3}\right).
\end{equation}

Then, we express the smoothness term $E^3_{\text{bdy}}$ as follows:
\begin{equation}
\begin{aligned}
    E^3_{\text{bdy}}(\mathbf{c},\mathbf{r},\mathbf{T}) = &  \sum^{k}_{i=1} \sum_{\mathcal{T}_l \in \mathcal{N}^{\mathcal{F}}(\gamma_i)} \left(  \frac{\cot{\theta_{l}^0}}{2}(\nu(v_{l}^1)-\nu(v_l^2))^2 \right.\\
    &+ \left. \frac{\cot{\theta_l^1}}{2}(\nu(v_l^2)-\nu(v_l^0))^2 +\frac{\cot{\theta_l^2}}{2}(\nu(v_l^0)-\nu(v_l^1))^2\right),
\end{aligned}
\end{equation}
where $\mathcal{T}_l = [v_{l}^0,v_{l}^1,v_{l}^2]$ is a triangle element on the $i$-th inner boundary, $\theta_l^0,\theta_l^1,\theta_l^2$ are the angles in $\mathcal{T}_l$, and $\nu(v_l^0)$, $\nu(v_l^1)$, $\nu(v_l^2)$ are the vertex Beltrami coefficients at $v_l^0$, $v_l^1$, $v_l^2$, which can be obtained by taking the average of the Beltrami coefficients at their 1-ring face neighborhood. Below, we denote $\nu(v_l^0)$, $\nu(v_l^1)$, $\nu(v_l^2)$ as $\nu_l^0, \nu_l^1, \nu_l^2$ for simplicity.   

The descent direction associated with $E^3_{\text{bdy}}$ for the $i$-th inner hole can be given by the three following terms:
\begin{equation}
\begin{aligned}
    dw^3_{i,1} = \frac{1}{n_i} \sum^{n_i}_{j=1} \sum_{\mathcal{T}_l \in \mathcal{N}^{\mathcal{F}}_{i,j}}&\left(\cot{\theta_l^0}(\nu_l^1 - \nu_l^2)\left(\frac{\partial \nu_l^1}{\partial r_{i,j}}-\frac{\partial \nu_l^2}{\partial r_{i,j}}\right) \right.\\
    &+ \cot{\theta_l^1}(\nu_l^2 - \nu_l^0)\left(\frac{\partial \nu_l^2}{\partial r_{i,j}} - \frac{\partial \nu_l^0}{\partial r_{i,j}}\right)\\
    &+\left. \cot{\theta_l^2}(\nu_l^0 - \nu_l^1)\left(\frac{\partial \nu_l^0}{\partial r_{i,j}} - \frac{\partial \nu_l^1}{\partial r_{i,j}}\right)\right),  
\end{aligned}
\end{equation}
\begin{equation}
\begin{aligned}
     dw^3_{i,2} = \frac{1}{n_i} \sum^{n_i}_{j=1} \sum_{\mathcal{T}_l \in \mathcal{N}^{\mathcal{F}}_{i,j}}&\left(\cot{\theta_l^0}(\nu_l^1 - \nu_l^2)\left(\frac{\partial \nu_l^1}{\partial T_{i,j}}-\frac{\partial \nu_l^2}{\partial T_{i,j}}\right)\right.\\
    &+ \cot{\theta_l^1}(\nu_l^2 - \nu_0)\left(\frac{\partial \nu_l^2}{\partial T_{i,j}} - \frac{\partial \nu_l^0}{\partial T_{i,j}}\right)\\
    &+\left.\cot{\theta_l^2}(\nu_l^0 - \nu_l^1)\left(\frac{\partial \nu_l^0}{\partial T_{i,j}} - \frac{\partial \nu_l^1}{\partial T_{i,j}}\right)\right),  
\end{aligned}
\end{equation}
\begin{equation}
\begin{aligned}
     dw^3_{i,3} = \frac{1}{n_i} \sum^{n_i}_{j=1} \sum_{\mathcal{T}_l \in \mathcal{N}^{\mathcal{F}}_{i,j}}&\left(\cot{\theta_l^0}(\nu_l^1 - \nu_l^2)\left(\frac{\partial \nu_l^1}{\partial c_i}-\frac{\partial \nu_l^2}{\partial c_i}\right)\right.\\
    &+ \cot{\theta_l^1}(\nu_l^2 - \nu_l^0)\left(\frac{\partial \nu_l^2}{\partial c_i} - \frac{\partial \nu_l^0}{\partial c_i}\right)\\
    &+\left.\cot{\theta_l^2}(\nu_l^0 - \nu_l^1)\left(\frac{\partial \nu_l^0}{\partial c_i} - \frac{\partial \nu_l^1}{\partial c_i}\right)\right). 
\end{aligned}
\end{equation}

Combining the above formulas for $dw^3_{i,1}, dw^3_{i,2}, dw^3_{i,3}$ for $i = 1, \dots, k$, the descent direction associated with the energy $E^3_{\text{bdy}}$ for all inner boundaries is given by:
\begin{equation}\label{eqt:dw3}
dw^3 = \sum^k_{i=1} \left( d w^3_{i,1} +   d w^3_{i,2} +  dw^3_{i,3}\right).
\end{equation}

Finally, combining Eq.~\eqref{eqt:dw1},~\eqref{eqt:dw2}, and~\eqref{eqt:dw3}, the overall descent direction for the domain alteration is given by:
\begin{equation}
    dw = dw^1 + \alpha dw^2 + \beta dw^3.
\end{equation}
We can then obtain the updated boundaries $\partial \mathcal{D}^1$ by
\begin{equation}\label{eqt:boundary}
     \partial \mathcal{D}^1 = \partial \mathcal{D}^0 + \delta t dw,
\end{equation}
where $\delta t$ is a prescribed time step size. We will later introduce an iterative scheme involving this updating step so that the boundaries become optimal over time.  

After obtaining $\partial \mathcal{D}^1$, we move on to the domain reconstruction step. Here, our goal is to obtain a quasiconformal map $\widetilde{g}$ such that $\mu(\widetilde{g})$ is close to $\nu$ and $\widetilde{g}(\partial \mathcal{D}^0) = \partial \mathcal{D}^1$. Also, it is desirable that the quasiconformal distortion of $\widetilde{g}$ is low so that the resulting mesh is sufficiently regular to serve as a good starting point for the subsequent density-equalization step. To achieve this, we compute the descent direction of the following energy involving the Beltrami coefficient:
\begin{equation}\label{eqt:reconstruct}
    E_R(\nu) = \alpha \int_{\mathcal{D}^0} |\nu|^2  + \beta \int_{\mathcal{D}^0} |\nabla \nu|^2.
\end{equation}
Here, the first term $\alpha \int_{\mathcal{D}^0} |\nu|^2$ is utilized to ensure that the quasiconformal distortion of the resulting map is small, while the second term $\beta \int_{\mathcal{D}^0} |\nabla \nu|^2$ is used to enhance the smoothness of the resulting map. Below, we denote the two integrals as $E_R^1(\nu) = \int_{\mathcal{D}^0}| \nu|^2$ and $E_R^2(\nu) = \int_{\mathcal{D}^0}|\nabla  \nu|^2$ for simplicity.

For $E_R^1(\nu) =  \sum^{|\mathcal{F} |}_{i=1} \text{Area}(\mathcal{T}_i)|\nu(\mathcal{T}_i)|^2$, where $\mathcal{T}_i$ is the $i$-th triangle element, the descent direction $d\nu^1$ can be computed by
\begin{equation}\label{eqt:ER1}
\begin{aligned}
    \left.\frac{d}{dt} \right|_{t = 0}E_R^1(\nu + \delta t\omega) &= \left.\sum^{|\mathcal{F} |}_{i=1} \text{Area}(\mathcal{T}_i) \frac{d}{d t}\right|_{t=0}\left|\nu(\mathcal{T}_i) + \delta t\omega_i \right|^2 \\
    & = 2\sum^{|\mathcal{F} |}_{i=1} \text{Area}(\mathcal{T}_i) \nu(\mathcal{T}_i) \omega_i.
\end{aligned}
\end{equation}
The derivative is negative when $\omega_i = - \nu(\mathcal{T}_i)$ for $i = 1,2,\dots,|\mathcal{F}|$. Hence, the descent direction of the Beltrami coefficient is $d\nu^1 = - \nu$. 

Using the same method, $E_R^2$ can be discretized on the triangle mesh:
\begin{equation}\label{eqt:E_R2}
\begin{aligned}
    E_R^2(\nu) = \sum^{|\mathcal{F} |}_{i=1} & \int_{\mathcal{T}_i}|\nabla \nu|^2 \\
     = \sum^{|\mathcal{F} |}_{i=1} & \left(\frac{\cot \theta_{i,0}}{2} (\nu(v_{i,1}) - \nu(v_{i,2}))^2 + \frac{\cot \theta_{i,1}}{2} (\nu(v_{i,2}) - \nu(v_{i,0}))^2\right.\\
     & \left.  + \frac{\cot \theta_{i,2}}{2} (\nu(v_{i,0}) - \nu(v_{i,1}))^2\right),
\end{aligned}
\end{equation}
where $\mathcal{T}_i = [v_{i,0},v_{i,1},v_{i,2}]$ is the $i$-th triangle element and $\theta_{i,0},\theta_{i,1},\theta_{i,2}$ are the three angles in it. Then we can obtain the increment $d\nu^2 =  \Delta \nu $ as well.

Therefore, the descent direction for the problem \eqref{eqt:reconstruct} is
\begin{equation}
    d\nu =  \alpha d\nu^1 + \beta d\nu^2.
\end{equation}
We can then obtain the updated Beltrami coefficient by:
\begin{equation}\label{eqt:update_nu}
\widetilde{\nu} = \nu + \delta t d\nu,
\end{equation}
where $\delta t$ is the time step size.
Using the LBS algorithm~\cite{lui2013texture}, we can compute the modification map $\widetilde{g}$,
and the new circular domain $\mathcal{D}^1$ is given by 
$ \mathcal{D}^1 = \widetilde{g}(\mathcal{D}^0)$. The proposed geometry modification method is summarized in Algorithm~\ref{alg:GM}.

\begin{algorithm}[htb]
  \caption{Geometry modification (GM)}
  \label{alg:GM}
  \begin{algorithmic}[1]
    \Require
      A circular domain $\mathcal{D}^0$ with density $\rho$ and initial Beltrami coefficient $\nu$.
    \Ensure
      A modification map $\widetilde{g}:\mathcal{D}^0 \to \mathcal{D}^1$, where $\mathcal{D}^1$ is a new circular domain.
    \State Compute the velocity $\mathbf{v} = -\frac{\nabla \rho^{\mathcal{V}}}{\rho^{\mathcal{V}}}$;
    \State Obtain the new boundaries $\partial \mathcal{D}^1$ based on $\mathbf{v}$ using Eq.~\eqref{eqt:dw1}--\eqref{eqt:boundary};
    \State Compute $\widetilde{\nu}$ using Eq.~\eqref{eqt:ER1}--\eqref{eqt:update_nu};
    \State Use the LBS method to obtain the modification map $\widetilde{g}$ from $\widetilde{\nu}$;  
    \State Obtain the new circular domain $\mathcal{D}^1 = \widetilde{g}(\mathcal{D}^0)$.   
  \end{algorithmic}
\end{algorithm}

\subsubsection{Beltrami density-equalizing descent}\label{sect:BDED}
From the above geometry modification step, we obtain a modified circular domain. However, the density may not be well-equalized in the modified domain. Therefore, our next step is to use the modified domain as the initial map and search for a suitable descent direction for density-equalization by solving the minimization problem~\eqref{eqt:DEQ_nu}.

We first discuss the energy $E^1_{\text{DEQ}} = \int |\nabla \rho |^2$ in Eq.~\eqref{eqt:DEQ_nu}, where $\rho = \rho_0$ is the density on $\widetilde{g}_1 \circ g_0(\mathcal{S})$. Here, we aim to express the descent direction for this energy as a descent direction of the Beltrami coefficient. By diffusion theory and Fick's law, the velocity induced by the density gradient is:
\begin{equation}
     dg = - \frac{\nabla \rho^{\mathcal{V}}}{\rho^{\mathcal{V}}}, 
\end{equation}
where the computation of $\nabla \rho^{\mathcal{V}}$ is described in Appendix~\ref{appendix:density}. To preserve the circular boundary shape obtained by Algorithm~\ref{alg:GM}, we enforce the Neumann boundary condition $\nabla \rho \cdot \mathbf{n} = 0$ as in~\cite{choi2018density} here, where $\mathbf{n}$ is the unit normal vector at each boundary. It is noteworthy that while the shape of each boundary is fixed, the boundary vertices are still allowed to slide along the boundary based on the density gradient to achieve a better density-equalizing effect.

It is easy to see that as $\widetilde{g}_1$ is perturbed, its associated Beltrami coefficient $\nu_0$ will be adjusted by $d\nu^1$. We can compute the adjustment via the Beltrami equation as follows:
\begin{equation}
    \frac{\partial (\widetilde{g}_1 + dg)}{\partial \Bar{z}} = (\nu_0 + d\nu^1)\frac{\partial (\widetilde{g}_1 + dg)}{\partial z},
\end{equation}
which implies
\begin{equation}
        \frac{\partial \widetilde{g}_1}{\partial \Bar{z}} + \frac{\partial dg}{\partial \Bar{z}} = \nu_0 \frac{\partial \widetilde{g}_1}{\partial z} +\nu_0 \frac{\partial dg}{\partial z} + d\nu^1 \frac{\partial \widetilde{g}_1}{\partial z} +d\nu^1 \frac{\partial dg}{\partial z}.
\end{equation}
Note that $\frac{\partial \widetilde{g}_1}{\partial \Bar{z}} = \nu_0 \frac{\partial \widetilde{g}_1}{\partial z}$. Hence, $d\nu^1$ can be obtained by
\begin{equation}\label{eqt:dnu1}
    d\nu^1 = \left(\frac{\partial dg}{\partial \Bar{z}} - \nu_0 \frac{\partial dg}{\partial z}\right) \left/ \frac{\partial (\widetilde{g}_1 + dg)}{\partial z} \right.. 
\end{equation}

Next, for the second term $E^2_{\text{DEQ}} =  \int |\nu_0|^2 =  \sum^{|\mathcal{F} |}_{i=1} \text{Area}(\mathcal{T}_i)|\nu_0(\mathcal{T}_i)|^2$, we compute the descent direction $d\nu^2$ by minimizing it iteratively. Analogous to Eq.~\eqref{eqt:ER1}, the Euler--Lagrange equation is given by
\begin{equation}
        \left.\frac{d}{dt} \right|_{t = 0}E^2_{\text{DEQ}}(\nu_0 + \delta t\omega) = 2\sum^{|\mathcal{F} |}_{i=1} \text{Area}(\mathcal{T}_i) \nu_0(\mathcal{T}_i) \omega_i.
\end{equation}
The derivative is negative when $\omega_i = - \nu_0(\mathcal{T}_i)$ for $i = 1,2,\dots,|\mathcal{F}|$. Hence, we obtain the descent direction of the Beltrami coefficient: $d\nu^2 = - \nu_0$. 

For the third part $E^3_{\text{DEQ}} =  \int|\nabla \nu_0|^2 =    \sum^{|\mathcal{F} |}_{i=1} \int_{\mathcal{T}_i}|\nabla \nu_0|^2$, we first discretize it on the triangle mesh as in Eq.~\eqref{eqt:E_R2}:
\begin{equation}
\begin{aligned}
    E^3_{\text{DEQ}}(\nu_0) =  &  \sum^{|\mathcal{F} |}_{i=1} \left(\frac{\cot \theta_{i,0}}{2} (\nu_0(v_{i,1}) - \nu_0(v_{i,2}))^2 + \frac{\cot \theta_{i,1}}{2} (\nu_0(v_{i,2}) - \nu_0(v_{i,0}))^2 \right.\\
    & \left. + \frac{\cot \theta_{i,2}}{2} (\nu_0(v_{i,0}) - \nu_0(v_{i,1}))^2\right),
\end{aligned}
\end{equation}
where $\mathcal{T}_i = [v_{i,0},v_{i,1},v_{i,2}]$ is the $i$-th triangle element and $\theta_{i,0},\theta_{i,1},\theta_{i,2}$ are the three angles in it. We can then obtain the increment $d\nu^3 =  \Delta \nu_0$.

Therefore, the descent direction for solving the optimization problem \eqref{eqt:DEQ_nu} is
\begin{equation}
    d\nu = d\nu^1 + \alpha d\nu^2 + \beta d\nu^3.
\end{equation}

Using the above descent direction formula, the updated Beltrami coefficient can be expressed as:
\begin{equation}\label{eqt:update_nu1}
\widetilde{\nu}_1 = \nu_0 + \delta t d\nu,
\end{equation}
where $\delta t$ is the time step size.

After obtaining $\widetilde{\nu}_1$, we apply the following cut-off function $\mathbb I$ to it:
\begin{equation}
\mathbb I(\widetilde{\nu}_1(z))= \begin{cases} \widetilde{\nu}_1(z) & \text { if } |\widetilde{\nu}_1(z)|< 1, \\ (1-\delta) \widetilde{\nu}_1(z) & \text { if } |\widetilde{\nu}_1(z)| \geq 1,\end{cases}
\end{equation}
where $\delta$ is a prescribed cut-off parameter. In practice, we set $\delta = 0.1$. We denote $\nu_1 = \mathbb I(\widetilde{\nu}_1)$. If $\|\nu_1\|_{\infty}$ is still larger than 1, we apply the cut-off function $\mathbb I$ repeatedly until $\|\nu_1\|_{\infty} <1$, which guarantees that the quasiconformal map associated with $\nu_1$ is bijective. Finally, we obtain a bijective deforming map $g_1:\mathcal{D}^0 \rightarrow \mathcal{D}^1$ by applying the LBS algorithm:
\begin{equation}
g_1 = \textbf{LBS}(\nu_1).
\end{equation}

We call this procedure \emph{Beltrami density-equalizing descent (BDED)} and summarize it in Algorithm~\ref{alg:BDED}.

\begin{algorithm}[htb]
  \caption{Beltrami density-equalizing descent (BDED)}
  \label{alg:BDED}
  \begin{algorithmic}[1]
    \Require
      A connected open triangulated surface $\mathcal{S}$, a population on each triangle of $\mathcal{S}$, a conformal flattening map $g_0: \mathcal{S} \to \mathcal{D}^0$, a modification map $\widetilde{g}_1: \mathcal{D}^0 \to \mathcal{D}^1$.
    \Ensure
      A map $g_1:\mathcal{D}^0 \rightarrow \mathcal{D}^1$.
    \State Compute the initial density function $\rho_0$ on $\widetilde{g}_1 \circ g_0(\mathcal{S})$;
    \State Compute the velocity $\mathbf{v}_0 = -\frac{\nabla \rho_0}{\rho_0}$;
    \State Initialize $\nu_0 = \mu(\widetilde{g}_1)$;
    \State Fixing the boundary shape, obtain $\widetilde{\nu}_{1}$ based on $\mathbf{v}_0$ and $\nu_0$ using Eq.~\eqref{eqt:dnu1}--\eqref{eqt:update_nu1};
    \State Obtain $\nu_{1} = \mathbb I(\widetilde{\nu}_{1})$;
    \State Use the LBS method to reconstruct $g_{1}$ from $\nu_{1}$; 
  \end{algorithmic}
\end{algorithm}

\begin{figure}[t]
    \centering
    \includegraphics[width=\textwidth]{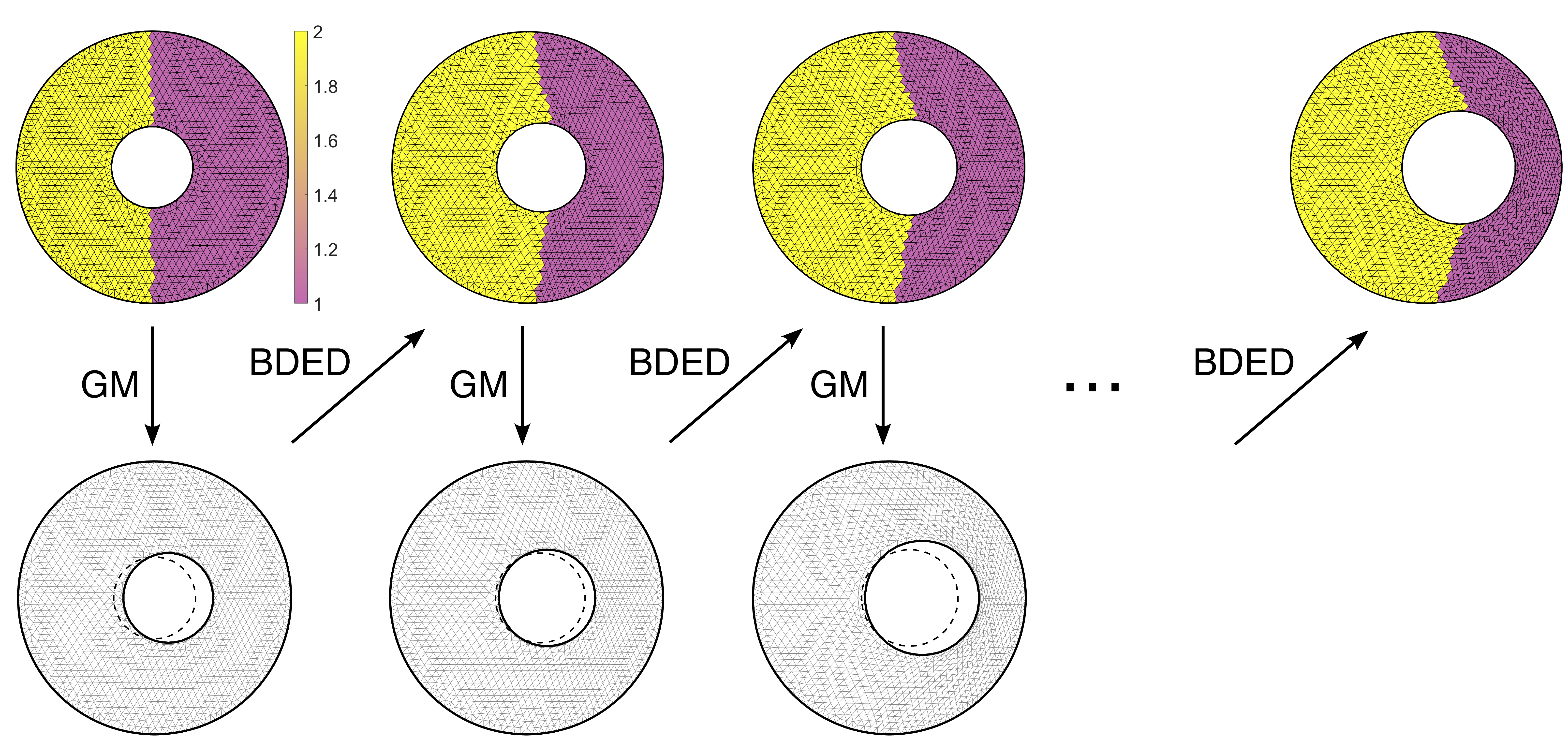}
    \caption{The proposed iterative scheme for computing density-equalizing quasiconformal (DEQ) maps. The geometry modification (GM) method and the Beltrami density-equalizing descent (BDED) method are applied iteratively to improve both the geometry of the target circular domain and the mapping result. Here, each triangle element is color-coded with the input population. The dashed curve in each GM step represents the previous inner boundary, and the solid curve represents its updated position.}
    \label{fig:illustration_deq}
\end{figure}

\subsubsection{DEQ algorithm}
In this section, we propose a method for computing the bijective density-equalizing quasiconformal (DEQ) maps.

For any given connected open triangle mesh $\mathcal{S}$, our goal is to find a quasiconformal mapping $f:\mathcal{S} \rightarrow \mathbb R^2$, such that $f(\mathcal{S})$ is a circular domain in $\mathbb R^2$ and the final density per unit area in $f(\mathcal{S})$ becomes a constant. Based on the GM and BDED steps, an iterative scheme is proposed to solve the optimization problem \eqref{eqt:DEQ_nu} (see Fig.~\ref{fig:illustration_deq}). More specifically, each iteration consists of two parts. Firstly, a modification function is computed for obtaining a new planar circular domain using the GM step. Secondly, a diffusion-based updated map is calculated for the new domain using the BDED step. We then repeat the GM and BDED steps to further improve both the domain geometry and the mapping until convergence.

It is natural to ask why alternating between GM and BDED is needed to achieve the optimal final mapping result. To explain this, note that in our approach we regard the diffusion process as a quasiconformal flow. In the GM step, the domain alteration is induced by the energy \eqref{eqt:modify_function} in the neighborhood of the inner boundaries. Then in the BDED step, we search for a diffusion-based descent direction with the new domain geometry prescribed. As the density at the regions away from the boundaries is not fully utilized in the GM step, the domain reconstructed and the descent direction result may not be optimal. By alternating between the GM and BDED steps, we can allow the density-equalization process to take place in a dynamically changing domain, thereby yielding an optimal circular domain and an optimal density-equalizing quasiconformal map.

The detailed procedure of the iterative scheme is as follows. Suppose $g_n$ is obtained at the $n$-th iteration, and $\mathcal{D}^n$ is the corresponding circular domain. The density function $\rho^{\mathcal{V}}_n$ can be calculated using the same method as above. Then, we compute the modification map $\widetilde{g}_{n+1}:\mathcal{D}^0 \to \mathcal{D}^{n+1}$ by Algorithm~\ref{alg:GM}, where $\mathcal{D}^{n+1}$ is the updated circular domain. Fixing the shape of $\mathcal{D}^{n+1}$, we then compute the updated map $g_{n+1}:\mathcal{D}^0 \to \mathcal{D}^{n+1}$ using the BDED step and obtain the updated Beltrami coefficient $\nu_{n+1} = \mu(g_{n+1})$. We repeat the iteration and stop when $\left|E_{\text{DEQ}}^n - E_{\text{DEQ}}^{n-1}\right| < \epsilon$ for some stopping parameter $\epsilon$, where $E_{\text{DEQ}}^n$ denotes the energy at the $n$-th iteration. The final density-equalizing quasiconformal map is then given by
\begin{equation}
    f = g_{N}  \circ g_0,
\end{equation}
where $N$ is the total number of iterations.

The proposed bijective density-equalizing quasiconformal (DEQ) mapping algorithm is summarized in Algorithm~\ref{alg:DEQ}.

\begin{algorithm}[htb]
  \caption{Bijective density-equalizing quasiconformal map (DEQ)}
  \label{alg:DEQ}
  \begin{algorithmic}[1]
    \Require
      A connected open triangulated surface $\mathcal{S}$, a population on each triangle, and a stopping parameter $\epsilon$.
    \Ensure
      A density-equalizing quasiconformal map $f:\mathcal{S} \rightarrow \mathcal{D}$, where $\mathcal{D}$ is an optimal circular domain.
        \State Compute an initial conformal flattening map $g_0:\mathcal{S} \rightarrow \mathcal{D}^0$, where $\mathcal{D}^0$ is a circular domain.
    \State Compute the initial density function $\rho_0$;
    \State Set $n$ = 0;
    \Repeat 
    \State Apply the GM method (Algorithm~\ref{alg:GM}) to get a modification map $\widetilde{g}_{n+1}:\mathcal{D}^0 \to \mathcal{D}^{n+1}$, where $\mathcal{D}^{n+1}$ is a new circular domain;
    \State Apply the BDED method (Algorithm~\ref{alg:BDED}) to obtain a map $g_{n+1}: \mathcal{D}^{0} \to \mathcal{D}^{n+1}$;     
        \State Update $\rho_{n+1}$ based on $g_{n+1}$;
        \State Update $n = n+1$;
    \Until $\left|E_{\text{DEQ}}^n - E_{\text{DEQ}}^{n-1}\right| < \epsilon$; \\
    \Return $f= g_{N} \circ g_0$, where $N$ is the total number of iterations;
  \end{algorithmic}
\end{algorithm}

We remark that by skipping the GM step in the DEQ iterative scheme, a shape-preserving density-equalizing quasiconformal map can be obtained, where the centers and radii of the circular holes will remain unchanged throughout the iterations. 

\subsection{Bijective landmark-matching density-equalizing quasiconformal map}\label{sect:LDEQ}
After establishing the DEQ method, we extend it to incorporate landmark-matching constraints based on the model~\eqref{eqt:LDEQ}. Let $\{ p_i\}^m_{i = 1}$ and $\{ q_i\}^m_{i = 1}$ denote the landmark vertices on $\mathcal{D}^0$ and their corresponding target positions, respectively. Analogous to the DEQ method, here we find an optimal planar circular domain $\mathcal{D}$ and an optimal Beltrami coefficient $\nu^*:\mathcal{D}^0\rightarrow \mathbb{C}$, minimizing the following variational model:
\begin{equation}\label{eqt:LDEQ_nu}
    E_{\text{LDEQ}}(g,\nu) = \int_{\mathcal{D}^0} |\nabla \rho(g) |^2 + \alpha \int_{\mathcal{D}^0} |\nu|^2 + \beta \int_{\mathcal{D}^0} |\nabla \nu|^2 
\end{equation}
subject to 
\begin{align}
&\|\nu\|_{\infty}  <1, \\
     &\nu = \mu(g),\\
     &g(p_i) = q_i, \ i = 1,\dots,m,
\end{align}
where $\mathcal{D}^0$ is the initial flatten circular domain, $g$ is a quasiconformal map, $\mu(g)$ is the Beltrami coefficient of $g$, and $\rho$ is the density function given by the population defined on $\mathcal{S}$. 
To achieve this, we develop an iterative scheme analogous to the DEQ algorithm, with some components of the BDED method modified to incorporate the landmark constraints.

More specifically, here we again start with an initial flattening map $g_0: \mathcal{S}\to \mathcal{D}^0$ and apply the GM method to obtain a new circular domain $\mathcal{D}^1$ and a modification map $\tilde{g}_1: \mathcal{D}^0 \to \mathcal{D}^1$. Then, we move on to the BDED step with the following changes. Instead of directly using the LBS method to solve the constraint optimization problem~\eqref{eqt:LDEQ_nu}, we follow the approach in~\cite{lam2014landmark} and solve the problem using the penalty splitting method. We consider the following energy functional:
\begin{equation}\label{eq:Land_penalty}
    E^{\text{split}}_{\text{LDEQ}}(\nu,\mu) = \int_{\mathcal{D}^0} |\nabla \rho(g) |^2 + \alpha \int_{\mathcal{D}^0} |\nu|^2 + \beta \int_{\mathcal{D}^0} |\nabla \nu|^2 + \eta \int_{\mathcal{D}^0} |\nu - \mu(g)|^2
\end{equation}
subject to the constraints that $\|\nu \|_{\infty} <1$ and $g(p_i) = q_i$ for $i = 1,\dots,m$. Here, note that we have included an extra term $\eta \int_{\mathcal{D}^0} |\nu - \mu(g)|^2$. Recall that the LBS method reconstructs a mapping $g$ based on a given Beltrami coefficient $\nu$. However, if landmark constraints are enforced, the resulting Beltrami coefficient $\mu(g)$ may not be close to the input $\nu$. In other words, the resulting mapping may not achieve the desired density-equalizing effect encoded by $\nu$. 
To resolve this issue, here we include the term $\eta \int_{\mathcal{D}^0} |\nu - \mu(g)|^2$ in the energy functional, where $\eta$ is a sufficiently large parameter. This term enforces that the Beltrami coefficient $\mu(g)$ is close to the solution $\nu$ of the model~\eqref{eqt:LDEQ_nu} even with the extra landmark constraints. The resulting mapping is then a bijective landmark-matching density-equalizing map. We then repeat the GM step and the modified BDED step until convergence. 

The detailed numerical procedure of the new approach is as follows. After getting the updated circular domain $\mathcal{D}^{n+1}$ and the modification map $\widetilde{g}_{n+1}: \mathcal{D}^{0} \to \mathcal{D}^{n+1}$ from the GM step, we compute the bijective landmark-matching density-equalizing map $g_{n+1}: \mathcal{D}^{0} \to \mathcal{D}^{n+1}$. Assume $\nu_n$ and $\mu_n$ are obtained at the $n$-th iteration. Fixing $\nu_n$, $\mu_{n+1}$ can be obtained by minimizing 
\begin{equation}\label{eqt:split1}
    E^{\text{split}}_{1}(\nu_n,\mu) = \sum^{|\mathcal{F}|}_{i = 1} \left(\int_{\mathcal{T}_i} |\nabla \rho^{\mathcal{F}}_n |^2 + \eta \int_{\mathcal{T}_i} |\nu_n - \mu|^2\right).
\end{equation}
By using the velocity $\mathbf{v}_n = \frac{\nabla \rho^{\mathcal{V}}_n}{\rho^{\mathcal{V}}_n}$ induced by the density function $\rho^{\mathcal{V}}_n$, we can get the adjustment to the Beltrami coefficient $d\mu^1$ of the first term. For the second term, the descent direction is given by $d\mu^2 = -2(\mu - \nu_n)$. Then, the final adjustment is
\begin{equation}
    d\mu = d\mu^1 + \eta d\mu^2.
\end{equation}
By the formula for the descent direction above, we can obtain the updated Beltrami coefficient
\begin{equation}\label{eqt:nu_np1_from_nu_n}
    \mu_{n+1} = \mu_n + \delta t d\mu,
\end{equation}
where $\delta t$ is the time step size.

Then, we follow the procedure in~\cite{lam2014landmark} and use LBS to compute $g_{n+1}^*:\mathcal{D}^0 \to \mathcal{D}^{n+1}$ such that $\mu(g_{n+1}^*) = \mu_{n+1}$ and $g_{n+1}^*$ matches all the landmarks constraints, and then update $\mu_{n+1}$ by $\mu_{n+1} = \mu(g_{n+1}^*)$. Once $\mu_{n+1}$ is obtained, we fix it and minimize the following energy
\begin{equation}\label{eqt:split2}
    E^{\text{split}}_{2}(\nu,\mu_{n+1}) = \sum^{|\mathcal{F}|}_{i = 1} \left(\alpha \int_{\mathcal{T}_i} |\nu|^2 + \beta \int_{\mathcal{T}_i} |\nabla \nu |^2 + \eta \int_{\mathcal{T}_i} | \nu - \mu_{n+1} |^2 \right)
\end{equation}
to get the updated Beltrami coefficient $\nu_{n+1}$. By considering the Euler--Lagrange equation, minimizing the above energy is equivalent to solving
\begin{equation}\label{eqt:nu_np1_from_mu_np1}
    (\alpha I- \beta \Delta + \eta I)\nu_{n+1} = \eta \mu_{n+1},
\end{equation}
where $\Delta = A^{-1}L$ (see Appendix~\ref{appendix:density} for the discretization details).

However, $\nu_{n+1}$ calculated above may not be associated with a landmark-matching density-equalizing diffeomorphism. Therefore, we use LBS with $\nu_{n+1}$ as the input together with the landmark constraints to obtain a landmark-matching quasiconformal map $g_{n+1}:\mathcal{D}^0 \to \mathcal{D}^{n+1}$, whose Beltrami coefficient $\mu(g_{n+1})$ closely resembles $\nu_{n+1}$. Then, using the direction $d = \mu(g_{n+1}) - \nu_{n+1}$ as the descent direction, we update $\nu_{n+1}$ by $\nu_{n+1} \leftarrow \nu_{n+1}+\delta t d$, where $\delta t$ is the time step size. Analogous to the DEQ algorithm, we repeat the above geometry modification and density-equalizing mapping processes until $\left|E_{\text{LDEQ}}^n - E_{\text{LDEQ}}^{n-1}\right| < \epsilon$ for some stopping parameter $\epsilon$, where $E_{\text{LDEQ}}^n$ denotes the energy at the $n$-th iteration. The proposed bijective landmark-matching density-equalizing quasiconformal (LDEQ) mapping algorithm is summarized in Algorithm~\ref{alg:LDEQ}.

\begin{algorithm}[htb]
  \caption{Bijective landmark-matching density-equalizing quasiconformal map (LDEQ)}
  \label{alg:LDEQ}
  \begin{algorithmic}[1]
    \Require
      A connected open triangulated surface $\mathcal{S}$, prescribed landmark correspondences $\{p_i\}^{m}_{i = 1} \leftrightarrow \{q_i\}^{m}_{i = 1}$, a population on each triangle, and a stopping parameter $\epsilon$.
    \Ensure
      A bijective landmark-matching density-equalizing quasiconformal map $f:\mathcal{S} \rightarrow \mathcal{D}$, where $\mathcal{D}$ is an optimal circular domain.
    \State Compute an initial conformal flattening map $g_0:\mathcal{S} \rightarrow \mathcal{D}^0$, where $\mathcal{D}^0$ is a circular domain.
    \State Compute the initial density function $\rho_0$;
    \State Set $n$ = 0;
    \State Initialize $\mu_0 = \nu_0 = 0$;
    \Repeat 
        \State Apply the GM method (Algorithm~\ref{alg:GM}) to get a modification map $\widetilde{g}_{n+1}:\mathcal{D}^0 \to \mathcal{D}^{n+1}$, where $\mathcal{D}^{n+1}$ is a new circular domain;
        \State Fixing the boundary shape, obtain $\mu_{n+1}$  using $\mu_{n}$ and $\nu_n$ by Eq.~\eqref{eqt:split1}--\eqref{eqt:nu_np1_from_nu_n};
        \State Use \text{LBS} to reconstruct $g_{n+1}^*:\mathcal{D}^0 \to \mathcal{D}^{n+1}$ based on $\mu_{n+1}$ together with the landmark constraints;
        \State $\mu_{n+1} \leftarrow \mu(g_{n+1}^*)$;
        \State Obtain $\nu_{n+1}$ using $\mu_{n+1}$ by Eq.~\eqref{eqt:nu_np1_from_mu_np1};
        \State Use \text{LBS} to reconstruct $g_{n+1}:\mathcal{D}^0 \to \mathcal{D}^{n+1}$ based on $\nu_{n+1}$ together with the landmark constraints;
        \State $\nu_{n+1} \leftarrow \nu_{n+1} + \delta t d$;
        \State Update $\rho_{n+1}$ based on $g_{n+1}$;
        \State Update $n = n+1$;
    \Until $\left|E_{\text{LDEQ}}^n - E_{\text{LDEQ}}^{n-1}\right| < \epsilon$; \\
    \Return $f= g_{N}  \circ g_0$, where $N$ is the total number of iterations;
  \end{algorithmic}
\end{algorithm}

\section{Experimental results}
\label{sec:experiments}
In this section, we present experimental results to demonstrate the effectiveness of our proposed DEQ and LDEQ algorithms. The algorithms are implemented using Matlab R2021a on the Windows platform. All experiments are conducted on a computer with an Intel(R) Core(TM) i9-12900 2.40 GHz processor and 32GB memory. All surfaces are discretized in the form of triangular meshes. In the following experiments, the stopping parameter is set to be $\epsilon = 10^{-2}$, and the time step size is set to be $\delta t = 0.1$.

\subsection{Bijective density-equalizing quasiconformal map}
We first test our proposed DEQ algorithm on synthetic and real data. In all experiments, we set $\alpha = 0.1$ and $\beta = 0.05$. The Koebe's iteration method~\cite{koebe1910konforme} is used for the initial flattening map.

We begin by presenting a synthetic example of mapping an annulus with four different density regions (Fig.~\ref{fig:annulus}(a)). The mapping result obtained by our proposed DEQ method is shown in Fig.~\ref{fig:annulus}(b), from which we can see that the position and size of the inner circular hole are optimized, yielding a smooth and bijective mapping result. As shown in the density histograms in Fig.~\ref{fig:annulus}(c)--(d), the DEQ method effectively equalizes the density. From the histogram of the norm of the Beltrami coefficient in Fig.~\ref{fig:annulus}(e), we can also see that the quasiconformal distortion is low. Moreover, from the energy plot in Fig.~\ref{fig:annulus}(f), we can see that the method converges quickly. 

    \begin{figure}[t!]
        \centering
       \includegraphics[width=\textwidth]{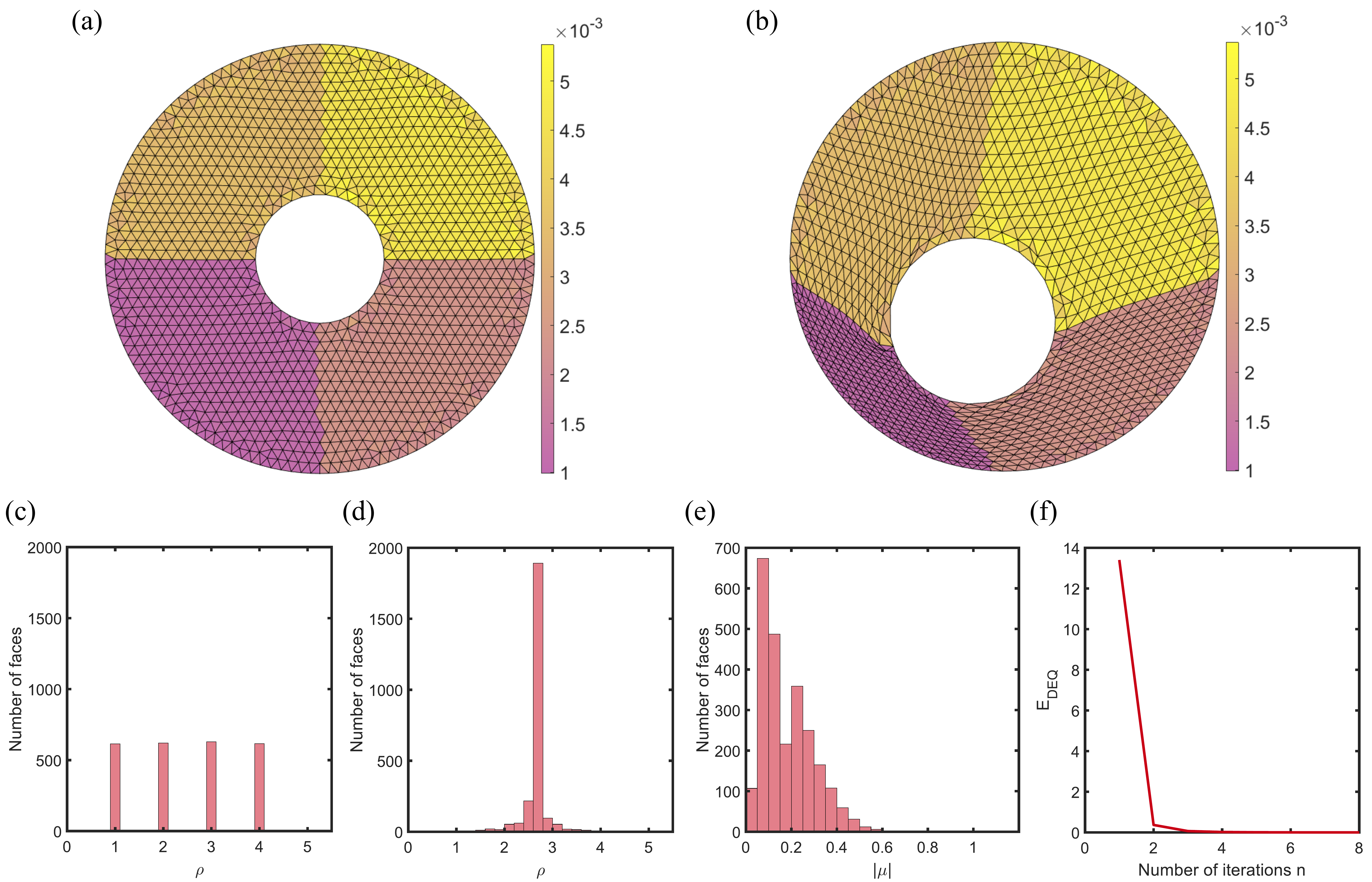}
        \caption{Density-equalizing quasiconformal map of an annulus with four different density regions. (a) The initial shape color-coded with the prescribed population. (b) The final DEQ map. (c) The histogram of the initial density on every triangle element. (d) The histogram of the final density. (e) The histogram of the norm of the Beltrami coefficient $|\mu|$. (f) The energy $E_{\text{DEQ}}$ throughout the iterations.}
        \label{fig:annulus}
    \end{figure}

It is natural to ask whether the geometry modification of the circular domain is important for achieving a good density-equalizing quasiconformal map. In Fig.~\ref{fig:2holes_compare}(a), we consider a multiply-connected domain with two holes and set the population to be two times larger at the middle part of the domain. Note that the prescribed population will naturally cause an expansion at the center of the domain, pushing the vertices to the two sides. Now, if we skip the GM step and directly compute a DEQ map with the domain shape fixed (i.e. a shape-preserving DEQ map), it can be observed in Fig.~\ref{fig:2holes_compare}(b) that the mapping result is with large distortion near the two holes. In particular, as the density gradient near the holes induces a velocity field that attempts to translate and enlarge the holes while they are enforced to be fixed, the triangle elements surrounding the two holes are largely distorted. By contrast, if we include the GM step, the DEQ mapping result becomes much more natural (Fig.~\ref{fig:2holes_compare}(c)), with the two holes moved to the two sides and enlarged. As shown in the density histograms in Fig.~\ref{fig:2holes_compare}(d)--(e), including the GM step gives a better density-equalization effect. We can also compare the norm of the Beltrami coefficients of the two mapping results as shown in Fig.~\ref{fig:2holes_compare}(f)--(g), from which we can see that the mapping result with GM yields a lower quasiconformal distortion. This demonstrates the significance of the GM step.

\begin{figure}[t]
    \centering
    \includegraphics[width=\textwidth]{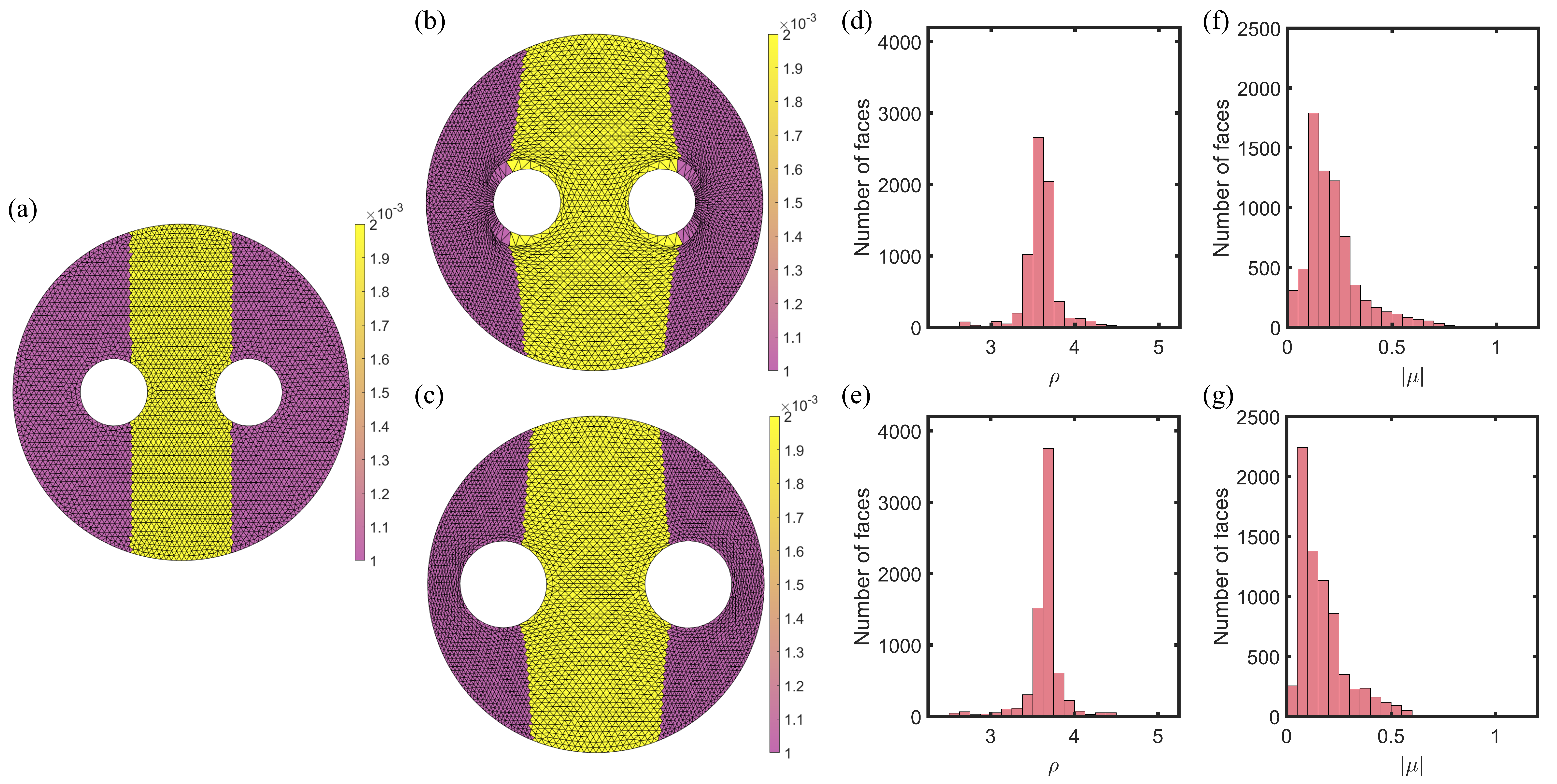}
    \caption{The importance of the geometry modification (GM) step in the computation of density-equalizing quasiconformal maps. (a) A multiply-connected domain with 2 holes, color-coded with the population prescribed on each triangle. (b) The DEQ mapping result without the GM step. (c) The DEQ mapping result with the GM step. (d) The final density histogram of DEQ without GM. (e) The final density histogram of DEQ. (f) The histogram of the norm of the Beltrami coefficient $|\mu|$ of DEQ without GM. (g) The histogram of $|\mu|$ of DEQ.}
    \label{fig:2holes_compare}
\end{figure}

After testing our proposed algorithm with the above synthetic examples, we consider computing a bijective density-equalizing flattening map of a human face model with 1 hole using our algorithm (Fig.~\ref{fig:1hole_face}(a)). Here, we set the population as the area of each triangle element on the original surface and apply the DEQ method in order to achieve an area-preserving parameterization. Fig.~\ref{fig:1hole_face}(b) and (c) show the initial flattening map and the final mapping result obtained by the DEQ algorithm, from which we can see that the mouth, eyes and eyebrows are shrunk in the initial map and restored in the final mapping result. The histogram of the initial density in Fig.~\ref{fig:1hole_face}(d) shows that the area distortion in the initial flattening map is large. By contrast, the final density highly concentrates at 1 as shown in Fig.~\ref{fig:1hole_face}(e), which indicates that the final bijective density-equalizing quasiconformal map effectively preserves the area ratio. We also consider the histogram of the norm of the Beltrami coefficient $|\mu|$ in Fig.~\ref{fig:1hole_face}(f), which shows that the quasiconformal distortion of the mapping result is also small. Fig.~\ref{fig:1hole_face}(g) shows the change of the energy $E_{\text{DEQ}}$ throughout the iterations, from which we can see that our method converges rapidly.

    \begin{figure}[t!]
        \centering
        \includegraphics[width=\textwidth]{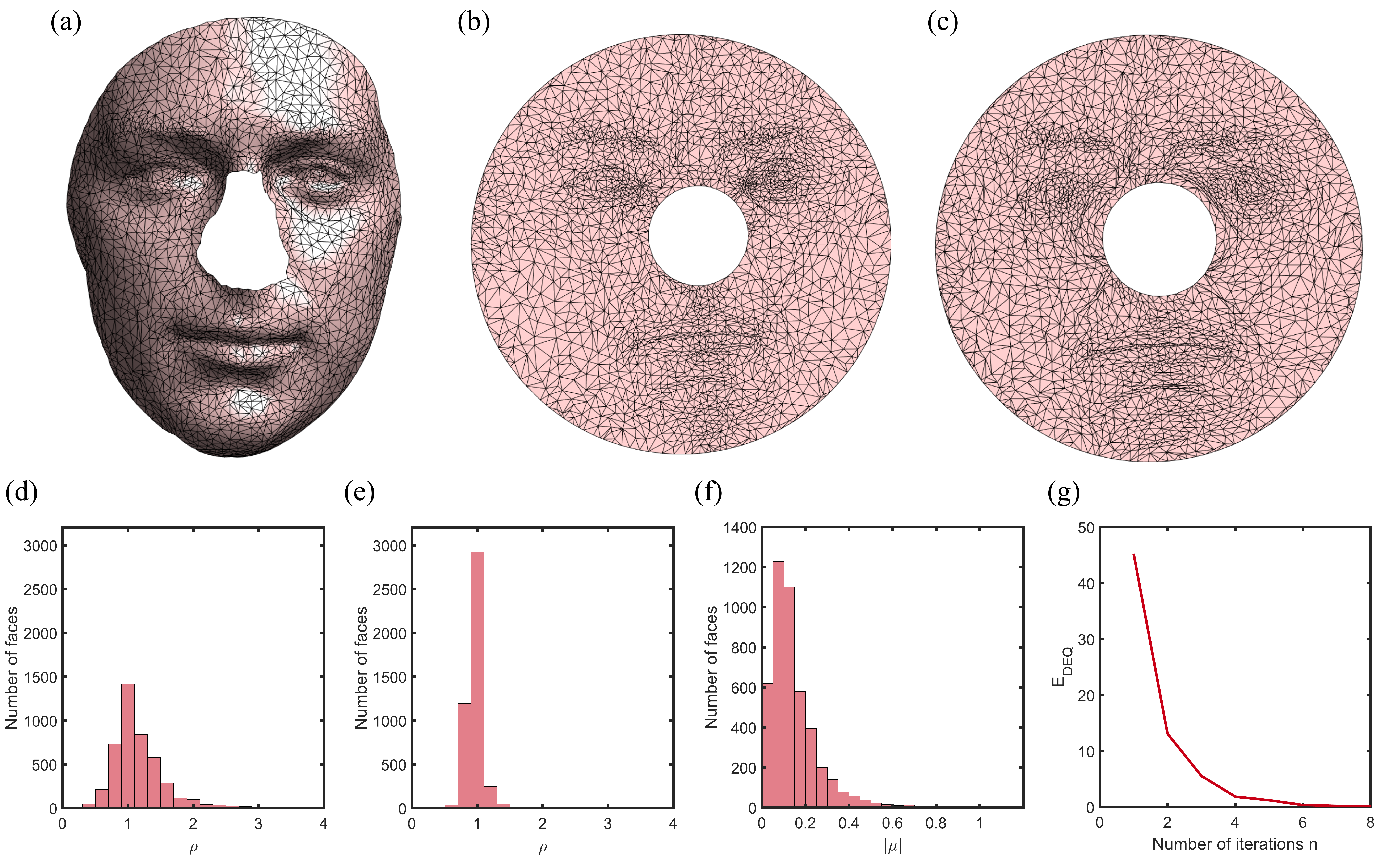}
        \caption{Bijective area-preserving parameterization of a human face with 1 hole. (a) The input surface. (b) The initial flattening map. (c) The final bijective density-equalizing quasiconformal map. (d) The histogram of the initial density on every triangle element. (e) The histogram of the final density. (f) The histogram of the norm of the Beltrami coefficient $|\mu|$. (g) The energy $E_{\text{DEQ}}$ throughout the iterations.}
        \label{fig:1hole_face}
    \end{figure}

Fig.~\ref{fig:2hole_face}(a) shows another human face with two holes located at the eyes. It can be observed that the nose is shrunk in the initial flattening map (Fig.~\ref{fig:2hole_face}(b)). By using the DEQ method with the triangle area as the prescribed population, we obtain an area-preserving parameterization (Fig.~\ref{fig:2hole_face}(c)) with the nose enlarged. Again, we can see that the density is effectively equalized (Fig.~\ref{fig:2hole_face}(d)--(e)) and the quasiconformal distortion is low (Fig.~\ref{fig:2hole_face}(f)). Also, the energy $E_{\text{DEQ}}$ is significantly reduced throughout the iterations (Fig.~\ref{fig:2hole_face}(g)). 

    \begin{figure}[t!]
        \centering
        \includegraphics[width=\textwidth]{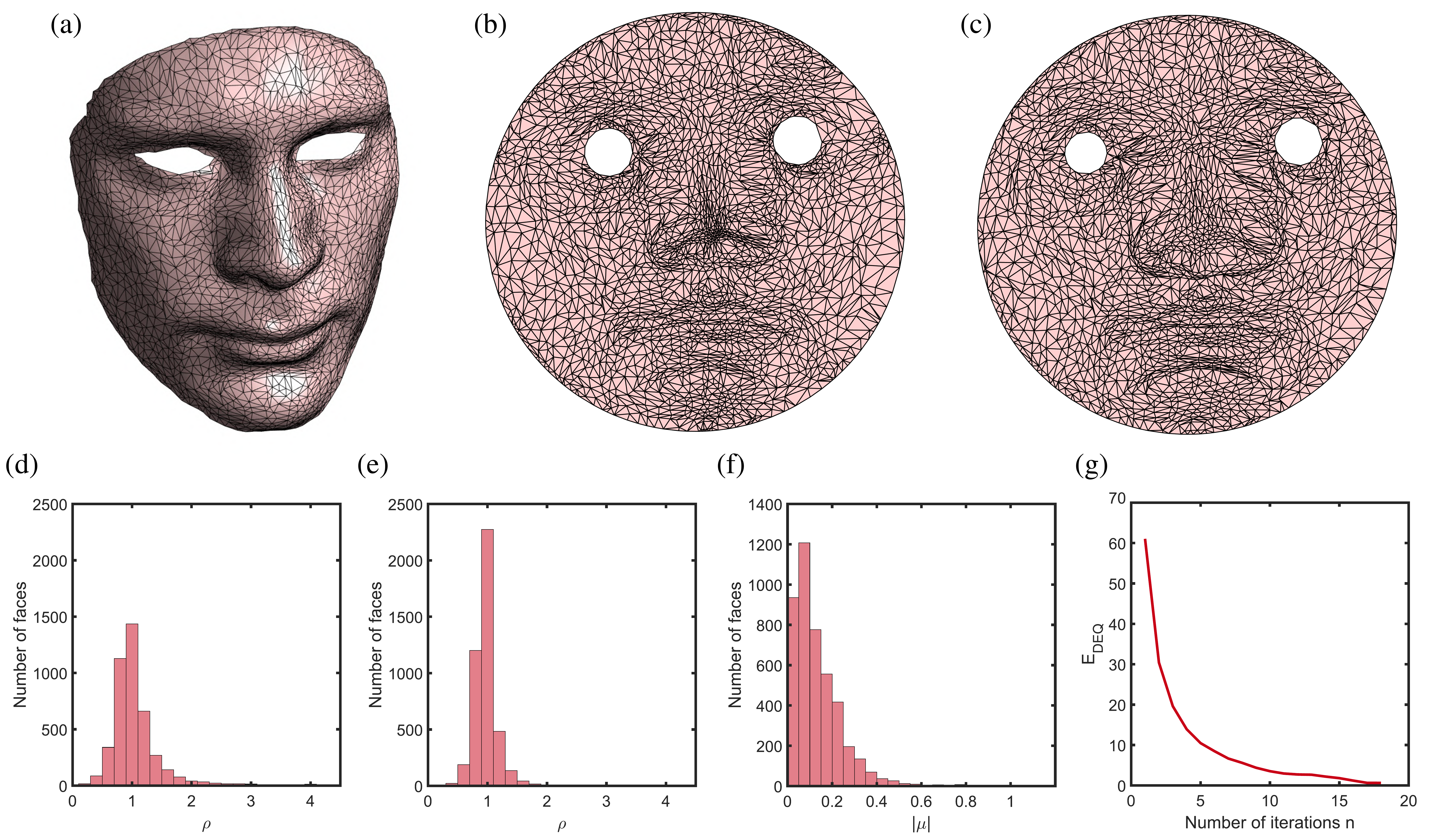}
        \caption{Bijective area-preserving parameterization of a human face with 2 holes. (a) The input surface. (b) The initial flattening map. (c) The final bijective density-equalizing quasiconformal map. (d) The histogram of the initial density on every triangle element. (e) The histogram of the final density. (f) The histogram of the norm of the Beltrami coefficient $|\mu|$. (g) The energy $E_{\text{DEQ}}$ throughout the iterations.}
        \label{fig:2hole_face}
    \end{figure}

It is noteworthy that our proposed algorithm can also be applied to simply-connected open surfaces. Specifically, we can skip the GM step and directly apply the BDED method iteratively (i.e. using the shape-preserving DEQ approach) to obtain a bijective density-equalizing disk parameterization. We can then compare the existing DEM method~\cite{choi2018density} and our DEQ method for mapping simply-connected open surfaces. For a fair comparison, the same stopping criterion is used in the two methods. Here, we consider flattening the human face in Fig.~\ref{fig:comparison_dem}(a) using the two methods, with the population set to be the area of each triangle element on the mesh except the mouth, and the population at the mouth set to be twice the area of the triangles there. As shown in the DEM result in Fig.~\ref{fig:comparison_dem}(b), the DEM method may lead to mesh fold-overs. By contrast, the proposed DEQ method preserves the bijectivity of the mapping and so there is no mesh overlap as shown in Fig.~\ref{fig:comparison_dem}(c). We can further compare the final density histograms (Fig.~\ref{fig:comparison_dem}(d)--(e)) and the histograms of the Beltrami coefficient $|\mu|$ (Fig.~\ref{fig:comparison_dem}(f)--(g)), from which we can see that the DEQ method achieves a comparable density-equalizing effect while significantly reducing the quasiconformal distortion.

    \begin{figure}[t]
        \centering
        \includegraphics[width=\textwidth]{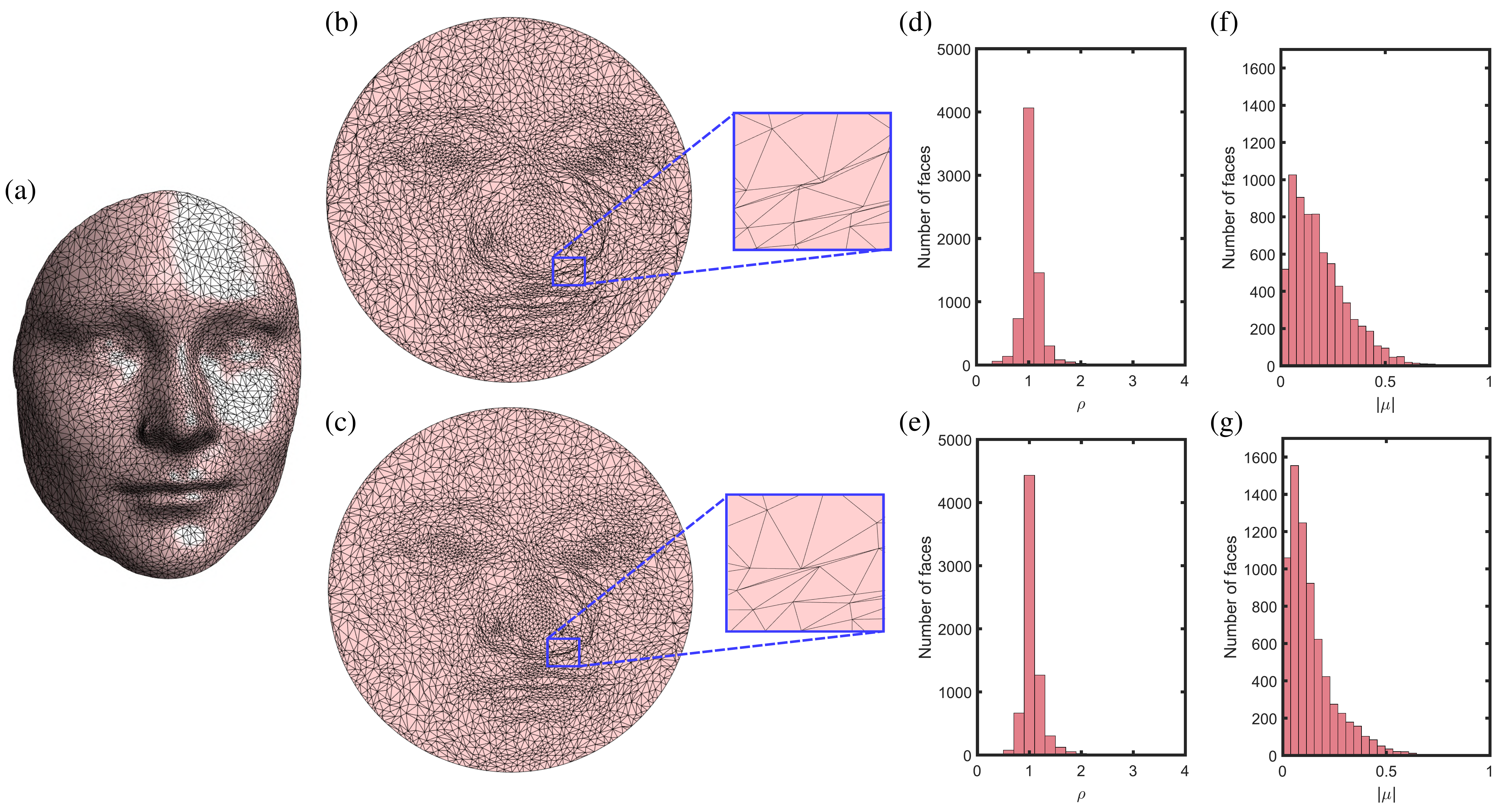}
        \caption{Comparison between DEM~\cite{choi2018density} and the proposed DEQ method for simply-connected open surfaces. (a) A simply-connected open face. (b) The DEM result with a zoom-in showing the presence of mesh overlaps. (c) The DEQ result with a zoom-in of the same region. (d) The final density histogram of DEM. (e) The final density histogram of DEQ. (f) The histogram of $|\mu|$ of DEM. (g) The histogram of $|\mu|$ of DEQ.}
        \label{fig:comparison_dem}
    \end{figure}

Fig.~\ref{fig:additional_results} shows several other examples of simply-connected and multiply-connected open surfaces and the DEQ mapping results, from which we can again see that the mappings are highly density-equalizing and bijective. 

    \begin{figure}[t!]
        \centering
        \includegraphics[width=0.9\textwidth]{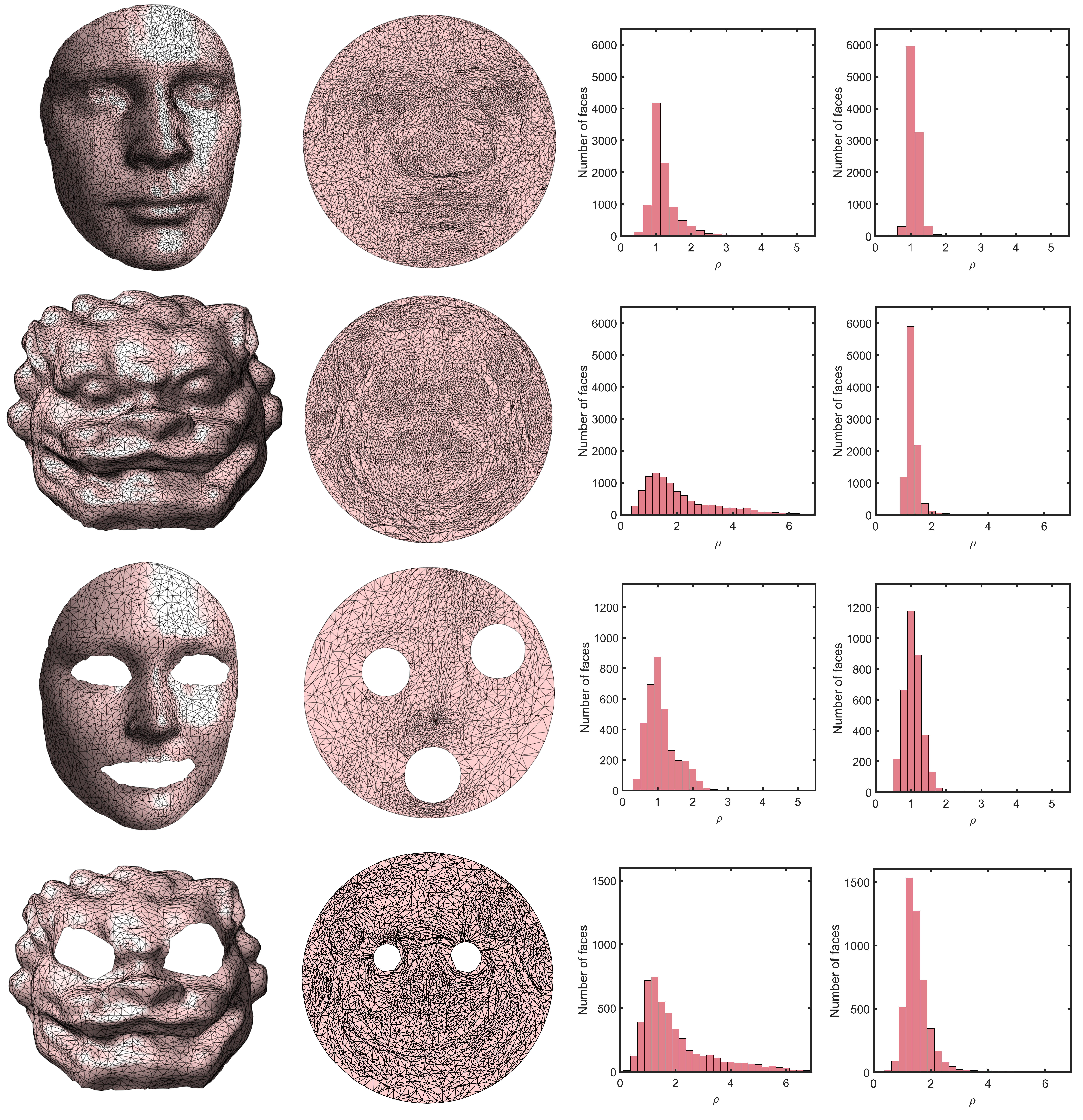}
        \caption{More examples of simply-connected and multiply-connected open surfaces and the DEQ mapping results. Each row shows one example. Left to right: The input surfaces, the DEQ mapping results, the histogram of the initial density, and the histogram of the final density.}
        \label{fig:additional_results}
    \end{figure}
    
To provide a more quantitative analysis of our DEQ algorithm, we present detailed statistics in Table \ref{tab:DEQ}, from which we can see that our method is highly efficient and accurate. Specifically, by considering the variance of the normalized final density, we can see that the mapping results are highly density-equalizing. Also, the mean value of the norm of the Beltrami coefficient shows that our method yields a low quasiconformal distortion. Moreover, as the bijectivity is guaranteed by our method, there is no mesh overlap in the mapping results.

\begin{table}[t!]
\footnotesize
    \caption{The performance of our DEQ algorithm. For each surface, we record the number of triangle elements, the time taken for the entire algorithm, the variance of the normalized final density $\widetilde{\rho} = \frac{\rho}{\text{Mean}(\rho)}$ where $\rho = \frac{\text{Given population}}{\text{Final area of triangle elements}}$, the mean value of the norm of the Beltrami coefficient $\mu$, and the number of overlaps. }\label{tab:DEQ}
  \begin{center}
  \begin{tabular}{|c|c|c|c|c|c|} \hline
    \bf Surface & \bf \# Faces & \bf Time(s) &\bf $\text{Var}(\widetilde{\rho})$  &\bf  Mean$(|\mu|)$ & \bf \# Overlaps \\\hline
    Lion (0 holes)& 9999 & 2.42 & 0.07 & 0.18 & 0 \\ \hline
    Alex (0 holes) & 10000 & 0.38 & 0.05 & 0.14 & 0 \\ \hline
    Square (0 holes) & 10368 & 0.65 & 0.02 & 0.24 & 0\\ \hline
    Annulus (1 hole) & 2479 & 0.25 & 0.10 & 0.17 & 0 \\ \hline 
    Rectangle (1 hole) & 3000 & 0.28 & 0.10 & 0.19  & 0\\ \hline
    Amoeba (1 hole) & 14255 & 2.12 & 0.06 & 0.27 & 0 \\ \hline
    Alex (1 hole) & 37999 & 1.23 & 0.03 & 0.11 & 0 \\ \hline
    Lion (2 holes) & 4900 & 0.43 & 0.12 &  0.25 & 0\\ \hline
    David (2 holes) & 48853 & 4.86 & 0.03 & 0.15  & 0\\ \hline
    Sophie (3 holes) & 4000 & 0.33 & 0.05 & 0.16 & 0 \\ \hline
    David (3 holes) & 47550 & 5.95 & 0.08 & 0.22  & 0\\ \hline
  \end{tabular}
\end{center}
\end{table}

Table~\ref{tab:Compare} presents a quantitative comparison between DEQ and DEQ without the GM step. As shown in the table, incorporating the GM step results in a better density-equalization effect and lower quasiconformal distortion. Moreover, it is noteworthy that including the GM step in the DEQ method significantly accelerates the computation. This can be explained by the fact that the GM step effectively simplifies the optimization problem, thereby reducing the number of iterations needed. Recall that our goal is to find a bijective density-equalizing diffeomorphism by solving Eq.~\eqref{eqt:DEQ_nu}. As the initial circular domain may not be suitable for the deformation, more iterations may be necessary to achieve the desired result. By contrast, with the GM step included, we can modify the domain based on the diffusion process and quasiconformal energy, thereby getting a more suitable domain for the deformation. Consequently, the iterative process is accelerated significantly with the use of GM.

\begin{table}[t]
\footnotesize
  \caption{The performance of the proposed DEQ algorithm compared with DEQ without the geometry modification (GM) step.}\label{tab:Compare}
\begin{center}
\begin{tabular}{|c|c|c|c|c|c|c|}
\hline
\multicolumn{1}{|c|}{ \multirow{2}*{\textbf{Surface}} }& \multicolumn{3}{c|}{\textbf{DEQ}}&\multicolumn{3}{c|}{\textbf{DEQ without GM}} \\
\cline{2-7}
\multicolumn{1}{|c|}{}&\bf Time(s) & \bf $\text{Var}(\widetilde{\rho})$  & \bf Mean$(|\mu|)$ &\bf Time(s) & \bf $\text{Var}(\widetilde{\rho})$  & \bf Mean$(|\mu|)$\\
\hline
Annulus (1 hole) &0.25  & 0.10 & 0.17 & 0.28 & 0.13 & 0.19\\ \hline 
Rectangle (1 hole) & 0.28 & 0.10 & 0.19 & 0.42 & 0.13 & 0.26\\ \hline 
Amoeba (1 hole) & 2.12 & 0.06 & 0.27 & 3.32 & 0.10 & 0.29\\ \hline 
Alex (1 hole) & 1.23 & 0.03 & 0.11 & 2.12 & 0.04 & 0.17 \\ \hline 
David (2 holes) & 4.86 & 0.03& 0.15 & 7.54 & 0.08 & 0.20 \\ \hline 
David (3 holes) & 5.95 & 0.08 & 0.22 & 8.25  & 0.10 &  0.28 \\ \hline
\end{tabular}
\end{center}
\end{table}

We further compare the performance of our DEQ method and the existing DEM method~\cite{choi2018density} for mapping simply-connected open surfaces, with the same stopping criterion used in the two methods. As shown in Table~\ref{tab:SS}, DEQ effectively reduces the quasiconformal distortion while achieving a comparable density-equalizing effect. It also outperforms DEM in terms of the bijectivity of the mappings. This demonstrates the effectiveness of the proposed DEQ method for not only multiply-connected open surfaces but also simply-connected open surfaces.

\begin{table}[t]
  \caption{Comparison between the proposed DEQ method and the DEM method~\cite{choi2018density} for simply-connected open surfaces.}\label{tab:SS}
\begin{center}
\resizebox{\textwidth}{!}{
\begin{tabular}{|c|c|c|c|c|c|c|c|c|c|}
\hline
\multicolumn{1}{|c|}{ \multirow{2}*{\textbf{Surface}} }& \multicolumn{1}{c|}{ \multirow{2}*{\textbf{\# Faces}}} & \multicolumn{4}{c|}{\textbf{DEQ}}&\multicolumn{4}{c|}{\textbf{DEM}} \\
\cline{3-10}
\multicolumn{1}{|c|}{}&\multicolumn{1}{c|}{}& \bf Time(s) & \bf $\text{Var}(\widetilde{\rho})$  & \bf Mean$(|\mu|)$ &\bf \# Overlaps &\bf Time(s) & \bf $\text{Var}(\widetilde{\rho})$  & \bf Mean$(|\mu|)$ & \bf \# Overlaps\\
\hline
Peaks & 4108 & 0.20 & 0.02 & 0.44  & 0 & 0.09 & 0.01 & 0.47  & 45  \\ \hline
Human face & 5000 & 1.35 & 0.07 & 0.13 & 0 & 1.01 & 0.05  & 0.21 & 128 \\ \hline
Lion & 9999 & 2.16 & 0.08 & 0.17  & 0 & 1.20 & 0.05 & 0.25  & 81 \\\hline
Alex & 10000 & 0.27 & 0.06  & 0.15 & 0 & 0.15 & 0.04 & 0.28 & 63 \\\hline
Square & 10368 & 0.74 & 0.02 & 0.23 & 0 & 0.52 & 0.01 & 0.27 & 22 \\ \hline
\end{tabular}
}
\end{center}
\end{table}

    \begin{figure}[t!]
        \centering
        \includegraphics[width=\textwidth]{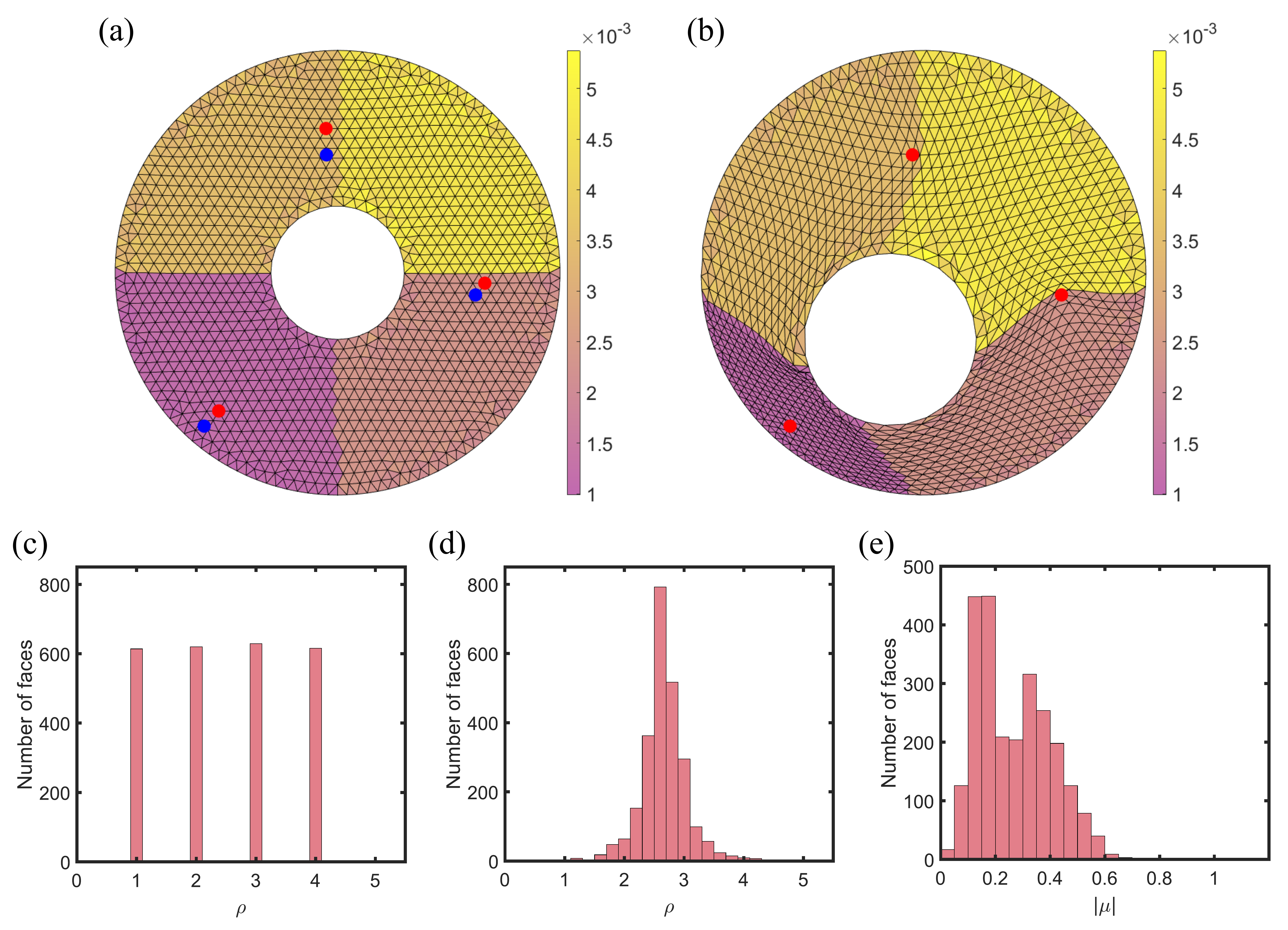}
        \caption{Bijective landmark-matching density-equalizing quasiconformal map of an annulus with 4 different density regions. (a) The initial shape color-coded with the prescribed population. The source landmarks and target positions are highlighted in red and blue respectively. (b) The LDEQ map with the red landmark vertices mapped to the target positions. (c) The histogram of the initial density. (d) The histogram of the final density. (e) The histogram of the norm of the Beltrami coefficient $|\mu|$. }
        \label{fig:annulus_LDEQ}
    \end{figure}

\subsection{Bijective landmark-matching density-equalizing quasiconformal maps}
Next, we test our proposed LDEQ algorithm for computing landmark-matching density-equalizing maps. In all experiments, we set $\alpha = 0.1$, $\beta = 0.05$ and $\eta = 10$. We first consider a synthetic example of an annulus in $\mathbb R^2$ with the input population and prescribed constraints shown in Fig.~\ref{fig:annulus_LDEQ}(a). Note that the input population is identical to the one used in the example in Fig.~\ref{fig:annulus}. It is easy to see that the LDEQ mapping result is bijective and all the corresponding landmarks are matched (Fig.~\ref{fig:annulus_LDEQ}(b)). Also, as shown in the histograms of the initial density, final density, and norm of the Beltrami coefficient (Fig.~\ref{fig:annulus_LDEQ}(c)--(e)), the density is well-equalized and the quasiconformal distortion is low. We can further compare the LDEQ mapping result with the DEQ mapping result in Fig.~\ref{fig:annulus}, from which we can clearly see the effect of the landmark constraints on the mapping result. 

We then consider computing the LDEQ mappings for multiply-connected open surfaces in $\mathbb{R}^3$. The first two rows in Fig.~\ref{fig:additional_results_LDEQ} show two examples of multiply-connected human face surfaces and the LDEQ mapping results. Here, the input population is the area of every triangle element, and landmark constraints are enforced to control the position of prominent features such as the mouth or the nose of the human faces. From the LDEQ mapping results, it can be observed that the landmark constraints are satisfied and the mappings are bijective. Comparing the initial and final density histograms, we can also see that the density is well-equalized under the LDEQ method.

    \begin{figure}[t!]
        \centering
        \includegraphics[width=\textwidth]{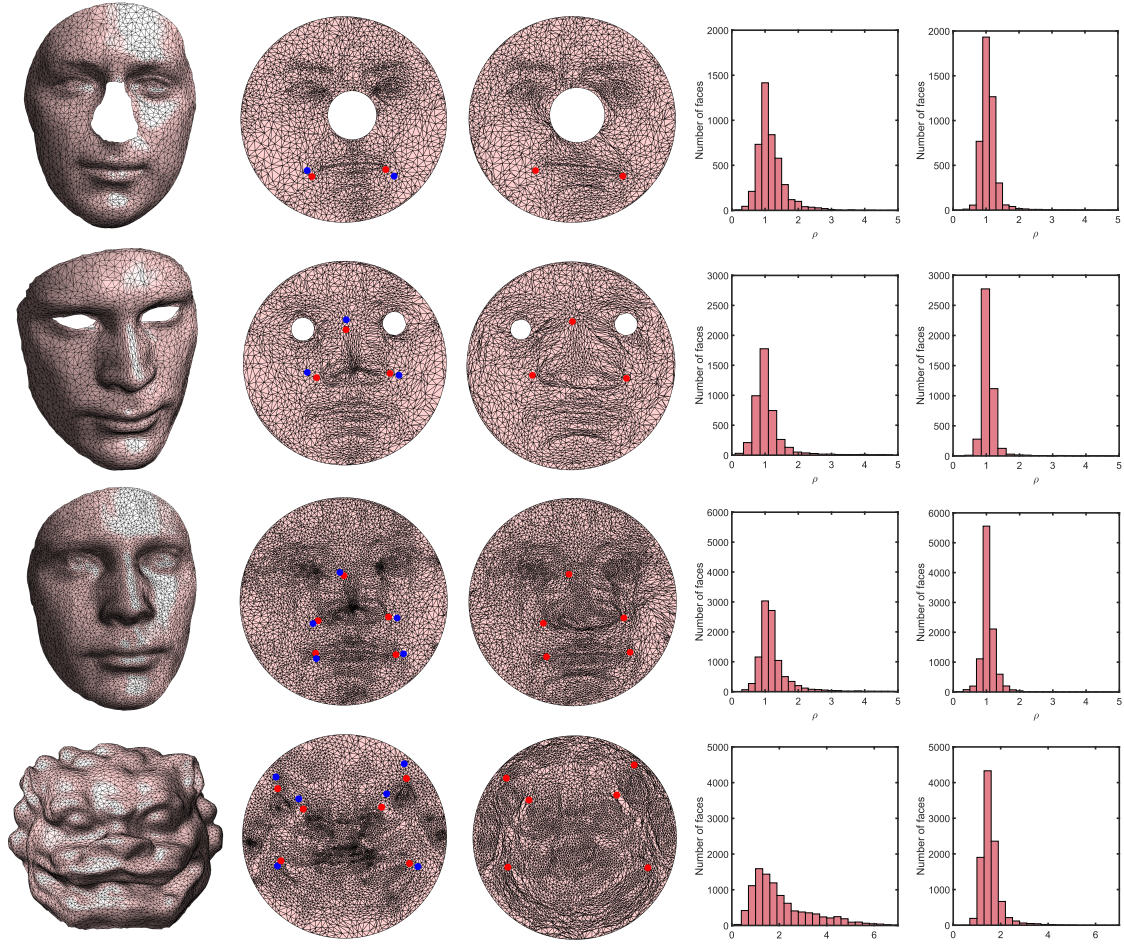}
        \caption{Examples of multiply-connected and simply-connected open surfaces and the LDEQ mapping results. Each row shows one example. Left to right: The input surface, the initial flattening map with the landmark constraints (red dots: source, blue dots: target), the LDEQ mapping result with the red dots mapped to the target position, the histogram of the initial density, and the histogram of the final density.}
        \label{fig:additional_results_LDEQ}
    \end{figure}

Analogous to the DEQ method, the LDEQ method can be applied to not only multiply-connected open surfaces but also simply-connected open surfaces. The last two rows in Fig.~\ref{fig:additional_results_LDEQ} show two examples of simply-connected open surfaces and the LDEQ mapping results. Again, it can be observed that the mapping results are landmark-matching and bijective, and the density is well-equalized.

For a more quantitative analysis, Table \ref{tab:LDEQ} shows the performance of our LDEQ algorithm. It is noteworthy that because of the landmark constraints, LDEQ results in larger $\text{sd}(\rho)$ and $\text{mean}(|\mu|)$ in general when compared to DEQ as in Table~\ref{tab:DEQ}. Also, the computation of LDEQ is slower than that of DEQ. Nevertheless, LDEQ can effectively incorporate landmark constraints and guarantee the bijectivity of the mapping results. 

\begin{table}[t]
\footnotesize
  \caption{The performance of our LDEQ algorithm. }\label{tab:LDEQ}
\begin{center}
  \begin{tabular}{|c|c|c|c|c|c|} \hline
    \bf Surface & \bf \# Faces & \bf Time(s) &\bf $\text{Var}(\widetilde{\rho})$ &\bf Mean$(|\mu|)$ & \bf \# Overlaps \\ \hline
    Alex (0 holes) & 10000 & 0.65 & 0.07 & 0.17  & 0\\ \hline
    Annulus (1 hole) & 2480 & 0.42 & 0.11 & 0.26 & 0\\ \hline 
    Sophie (1 hole) & 4499 & 0.35 & 0.06 & 0.16 & 0 \\ \hline
    David (2 holes) & 4400 & 0.75 & 0.06 & 0.18 & 0 \\ \hline    
    Lion (2 holes) & 4900 & 0.63 & 0.12 & 0.25 & 0 \\ \hline
  \end{tabular}
\end{center}
\end{table}

\section{Applications}\label{sec:application}
In this section, we present several applications of our proposed bijective density-equalizing quasiconformal mapping method.  

\subsection{Surface remeshing}
Our proposed method can be applied to surface remeshing. Given an open surface $\mathcal{S}$ in $\mathbb{R}^3$, the goal is to construct a new mesh structure on the surface with better quality. To achieve this, we can first prescribe a population on the surface and apply our proposed algorithm to compute a bijective DEQ map $f:\mathcal{S} \rightarrow \mathcal{D}$ onto a planar circular domain. Then, we can construct a regular mesh structure $\mathcal{M}$ on $\mathcal{D}$. Finally, using inverse mapping $f^{-1}$, the regular mesh structure $\mathcal{M}$ can be mapped back onto the surface $\mathcal{S}$ and we obtain the remeshed surface $f^{-1}(\mathcal{M})$. 

It is noteworthy that the effect of the DEQ map is closely related to the initial population distribution. Assigning larger populations to specific regions can increase the level of detail in those areas, as they are enlarged in the DEQ mapping result and more points of $\mathcal{M}$ are mapped onto those parts. Consequently, $f^{-1}$ maps more points back onto the corresponding regions of $\mathcal{S}$, resulting in the display of more details in those domains. Also, we can further preserve certain desired features of $\mathcal{S}$ during the remeshing process by using the proposed LDEQ method. Specifically, we can first place landmarks around the features and then compute the LDEQ mapping $g: \mathcal{S} \rightarrow \mathcal{D}$ with the landmarks mapped to their corresponding target positions. A regular mesh $\mathcal{M}$ is then constructed on $\mathcal{D}$ using DistMesh~\cite{persson2004simple}, and finally $\mathcal{M}$ is interpolated onto $\mathcal{S}$ using the inverse mapping $g^{-1}$. The resulting mesh $g^{-1}(\mathcal{M})$ is then a feature-preserving remeshed representation of $\mathcal{S}$.

    \begin{figure}[t!]
        \centering
        \includegraphics[width=\textwidth]{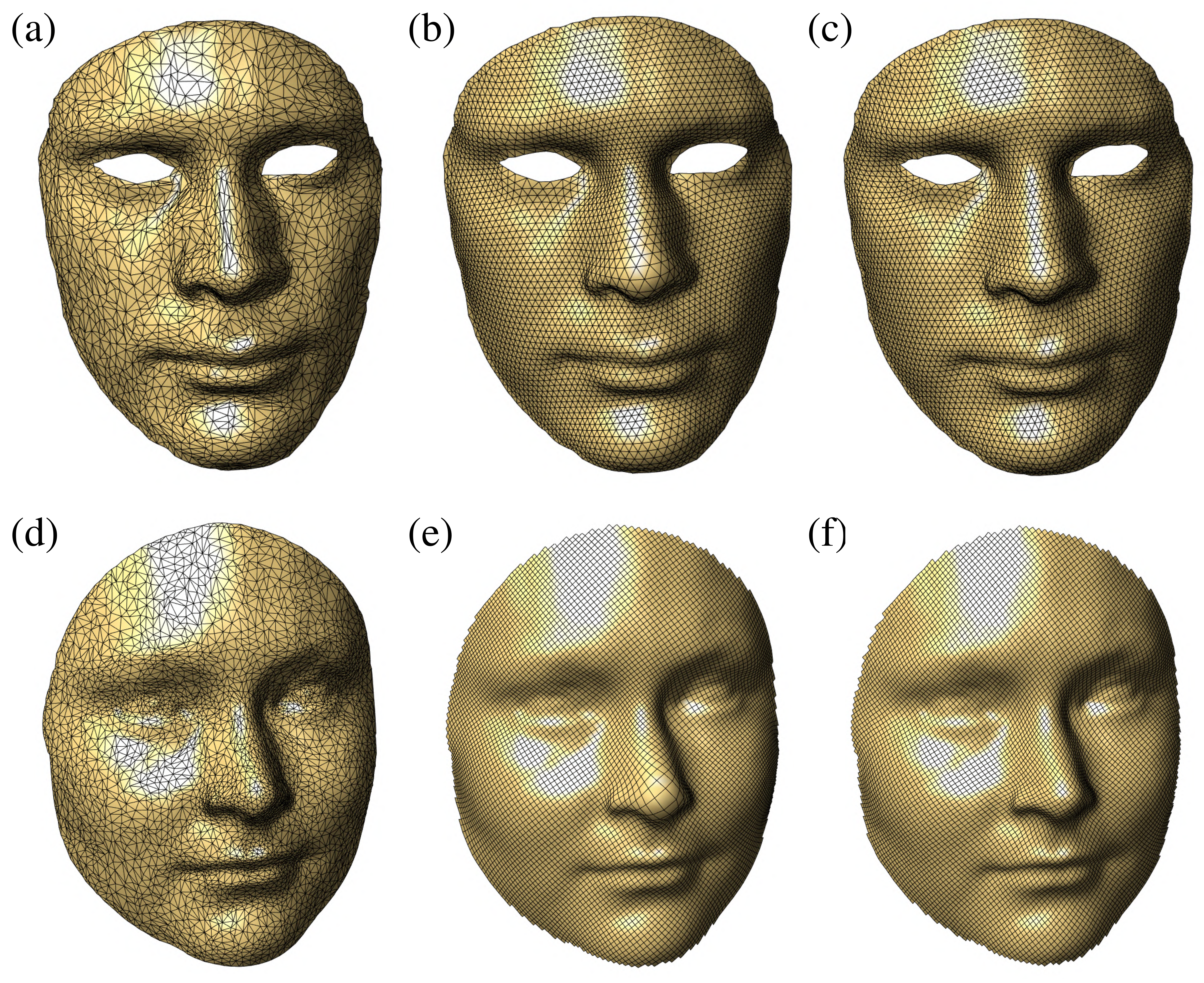}
        \caption{Remeshing human face models. (a) A multiply-connected human face mesh with 2 holes. (b):  The remeshing result via conformal parameterization. (c) The remeshing result via the proposed method. (d) A simply-connected human face mesh. (e) The quad-remeshing result via conformal parameterization. (f) The quad-remeshing result via the proposed method.}
        \label{fig:remeshing_face}
    \end{figure}

Fig.~\ref{fig:remeshing_face} shows two examples of remeshing a human face model using the DEQ method. We first consider a multiply-connected human face with 2 holes (Fig.~\ref{fig:remeshing_face}(a)), in which we can see that the initial triangulation is highly irregular. Fig.~\ref{fig:remeshing_face}(b) shows the remeshing results obtained via conformal parameterization. It can be observed that the triangulations are highly nonuniform at the nose in the conformal remeshing results, which can be explained by the fact that the conformal method does not control the area distortion of the parameterization. By contrast, using our proposed DEQ method with the population set to be proportional to the triangle area of the original mesh, we can achieve a much better remeshing result with a more uniform distribution (Fig.~\ref{fig:remeshing_face}(c)). As discussed previously, our method can also be applied to simply-connected open surfaces. In the second example, we consider remeshing a simply-connected human face (Fig.~\ref{fig:remeshing_face}(d)) using quad meshes. Again, we can see in the quad-remeshing result in Fig.~\ref{fig:remeshing_face}(e) that the conformal method leads to a non-uniform distribution, while the result obtained via our DEQ method is much more uniform (Fig.~\ref{fig:remeshing_face}(f)).

We then consider another experiment of remeshing a multiply-connected car model (Fig.~\ref{fig:remeshing_car}(a)) using the LDEQ method. It can be observed that the remeshing result via conformal parameterization (Fig.~\ref{fig:remeshing_car}(b)) loses the prominent features of the car hood and the headlamps. Also, the triangle density of the remeshed surface is highly uneven. By contrast, using our proposed algorithm, we can achieve a much more balanced distribution of points. Here, landmarks are placed at the car hood and the headlamps in the mapping computation to ensure the preservation of the prominent features, yielding the remeshing result in Fig.~\ref{fig:remeshing_car}(c). 

The above experiments show that our method is more advantageous than the conformal method for surface meshing.
 
    \begin{figure}[t!]
        \centering
        \includegraphics[width=\textwidth]{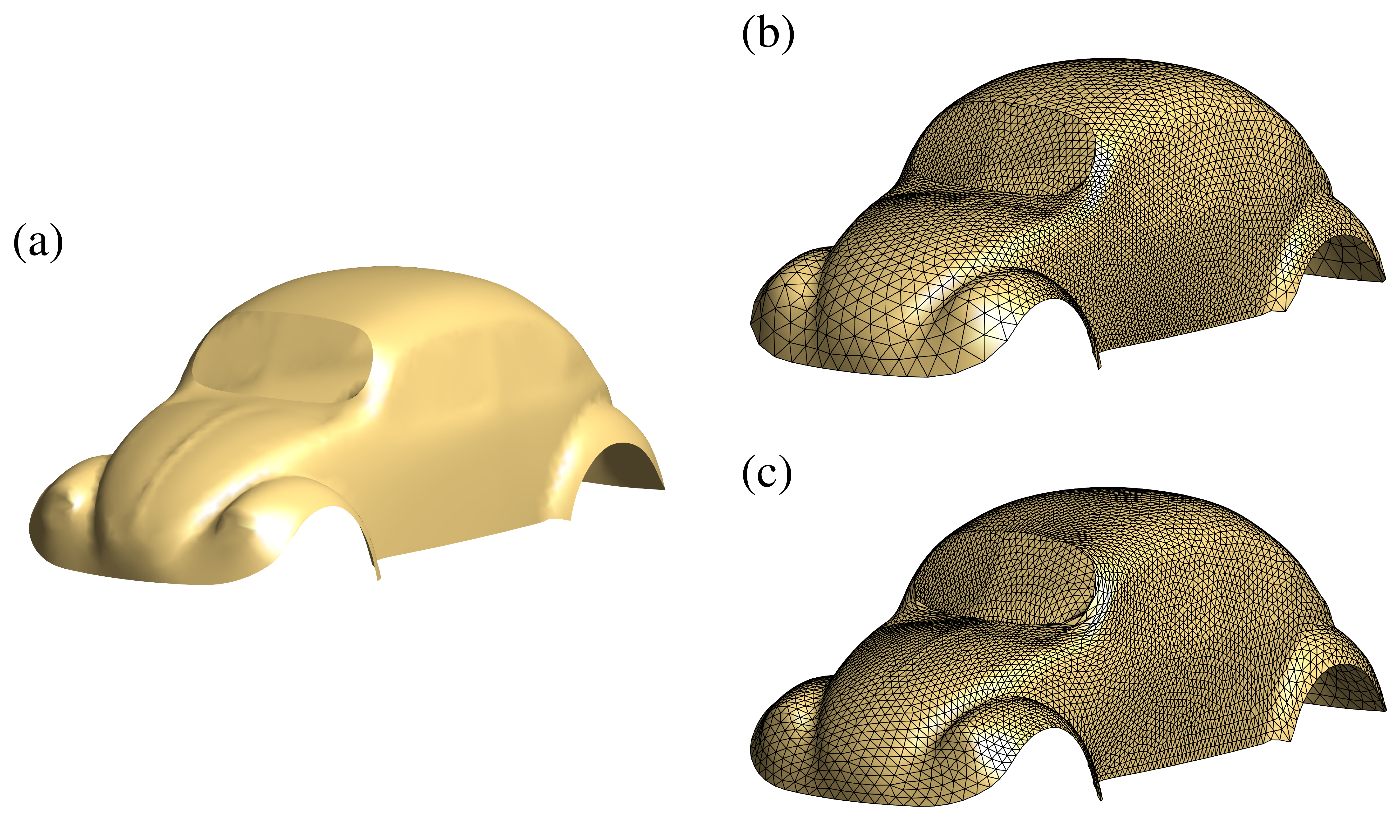}
        \caption{Remeshing a car model. (a) The original car surface. (b) The remeshing result via conformal parameterization. (c) The remeshing result via the proposed method.}
        \label{fig:remeshing_car}
    \end{figure}

\subsection{Texture mapping}
Our algorithm can be used for texture mapping on connected open surfaces. Specifically, after flattening a given 3D surface onto a planar circular domain, we can design a texture on the circular domain freely. We then use the inverse mapping to map the texture back onto the original surface. 

Note that conformal parameterizations are widely used for texture mapping because they can preserve the local geometry. However, the area distortion of the designed texture is large. When compared with the conformal method, our DEQ algorithm can preserve the area if we set the initial population as the area of the original surface. Moreover, since our method is obtained by solving Eq.~\eqref{eqt:DEQ}, the angle distortion is also small while the density of the unit area becomes constant.

Fig.~\ref{fig:texture} shows two examples of our method applied to different surfaces. A simply-connected human face and a multiply-connected human face are shown in Fig.~\ref{fig:texture}(a) and (d) respectively. We first design different textures and map them back onto the surfaces using conformal parameterization (Fig.~\ref{fig:texture}(b) and (e)). It can be observed that the right angles in the checkerboard pattern are well-preserved on the human face. However, some area distortions are visible around the nose and chin, indicating that the conformal method does not work well for those regions. By contrast, by applying the DEQ method, we can obtain texture mapping results with the checkerboard patterns uniformly mapped onto the faces (Fig.~\ref{fig:texture}(c) and (f)), which indicates the area is not distorted during the process. Moreover, it can be observed that orthogonality is preserved in most regions of the human face, which is attributed to the inclusion of the quasiconformality terms in the DEQ energy minimization model. The examples show that our DEQ method can achieve density-equalizing texture mapping with minimal angle distortion.  

    \begin{figure}[t!]
        \centering
        \includegraphics[width=0.95\textwidth]{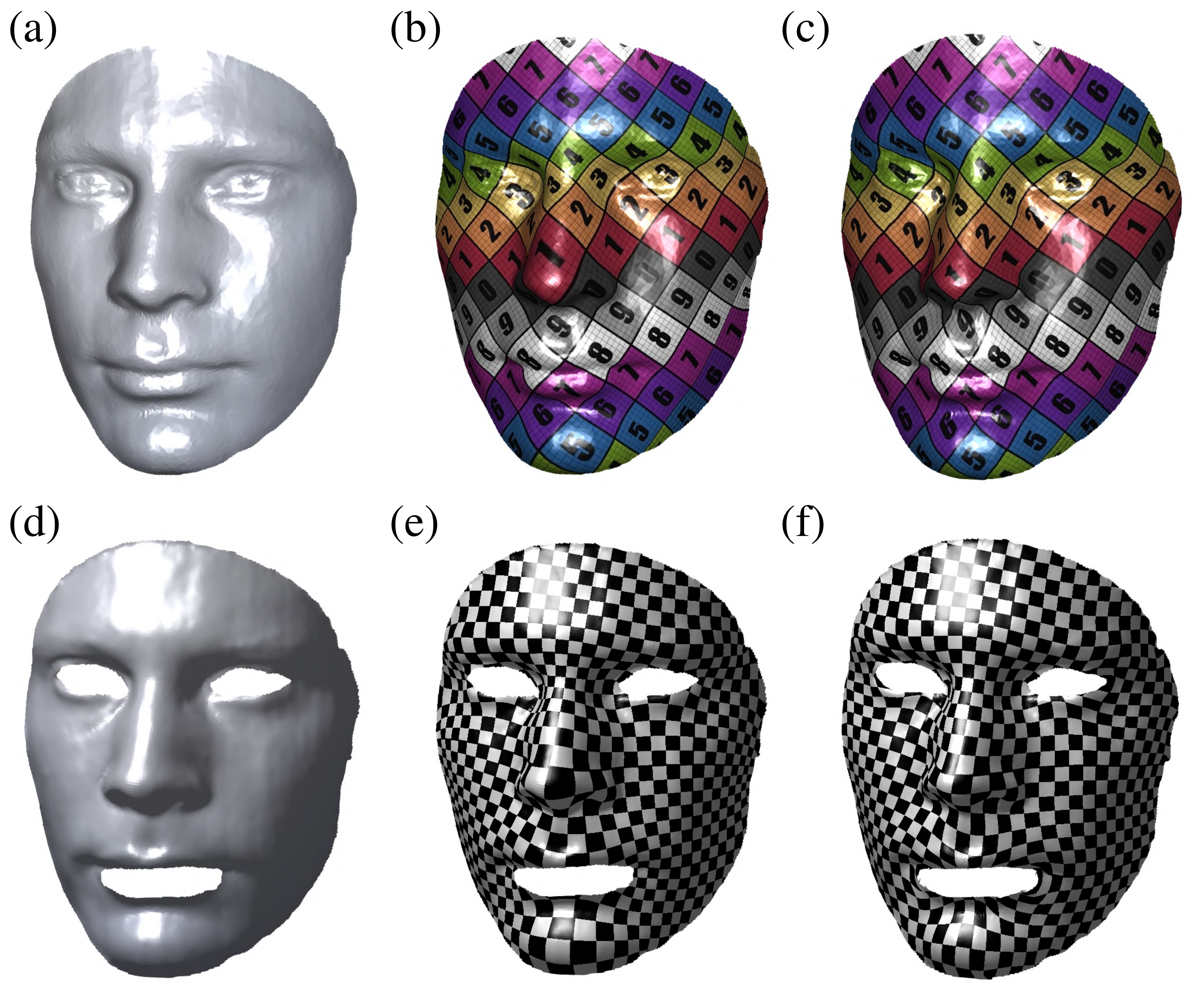}
        \caption{Texture mapping using the conformal and DEQ parameterizations. (a) A simply-connected human face. (b) Texture mapping using conformal parameterization. (c) Texture mapping using the DEQ method. (d) A multiply-connected human face. (e) Texture mapping using conformal parameterization. (f) Texture mapping using the DEQ method.} 
        \label{fig:texture}
    \end{figure}

\begin{figure}[t!]
    \centering
    \includegraphics[width=\textwidth]{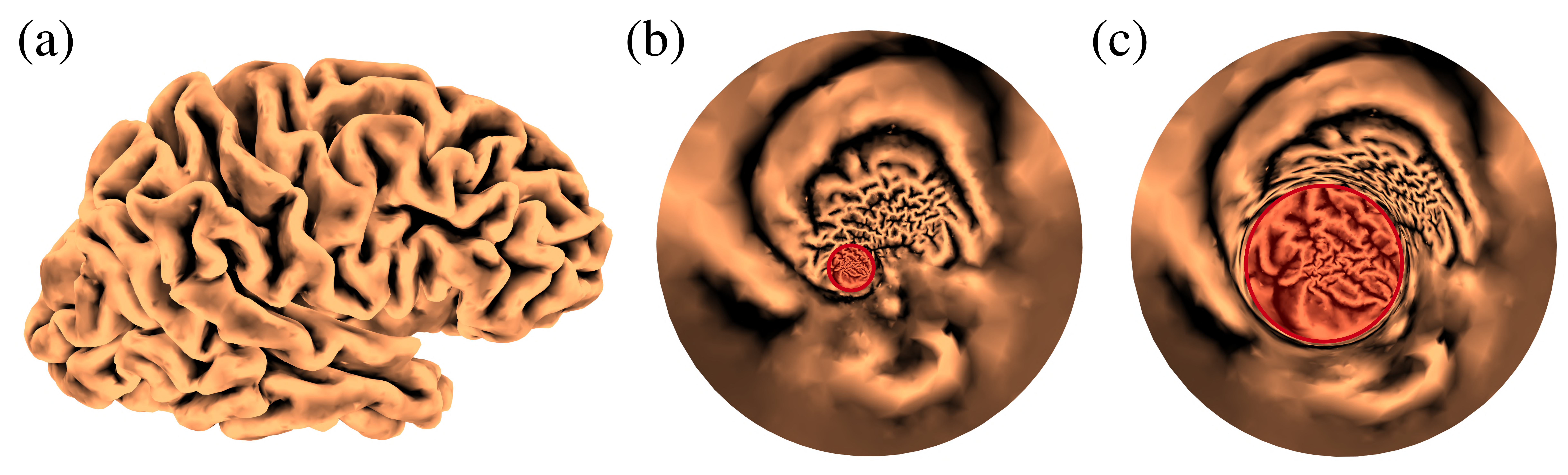}
    \caption{Application of our proposed method to medical visualization. (a) A highly convoluted human cortical surface. (b) Conformal parameterization of the surface onto the unit disk, where gyri and sulci are highly squeezed in a small region (highlighted in red). (c) The highlighted region is effectively enlarged using the DEQ method.}
    \label{fig:medical_visualization}
\end{figure}

\subsection{Medical visualization}
Our method can also be applied to medical visualization as an interactive shape magnifier. Given a complicated anatomical structure such as a human cortical surface (Fig.~\ref{fig:medical_visualization}(a)), one common approach for achieving a better visualization of it is to flatten it onto a simpler domain. However, conformal parameterization of such surfaces will easily produce a flattening result with certain parts highly squeezed in a small region (such as the sulci and gyri in Fig.~\ref{fig:medical_visualization}(b)), as the area distortion of conformal mappings is uncontrolled. This significantly hinders the visualization and shape analysis of such medical structures. 

To overcome this issue, we can apply our DEQ method with a suitable population prescribed. More specifically, the user can first prescribe a larger population at a region of interest (ROI). Then, by computing a DEQ map, we can effectively enlarge the ROI while preserving the bijectivity of the parameterization result (Fig.~\ref{fig:medical_visualization}(c)). Moreover, as the quasiconformal distortion is also taken into consideration in the DEQ model, the local geometry of the ROI is well-preserved under the mapping.

\section{Discussion}\label{sec:conclusion}
In this paper, we proposed a novel method for computing bijective density-equalizing quasiconformal mappings for both simply-connected and multiply-connected open surfaces. The main strategy is to minimize an energy functional that comprises a density term and two quasiconformality terms, where the density term measures the density distribution and the quasiconformality terms control the quasiconformal distortion and smoothness of the mapping. By minimizing the energy functional, our method can obtain a density-equalizing parameterization onto a 2D circular domain for any given connected open surface, with the bijectivity preserved and the quasiconformal distortion controlled. Landmark constraints can also be included in the mapping computation. Besides the applications to surface remeshing, texture mapping and medical visualization demonstrated in our paper, the proposed approach can be potentially combined with other mapping methods for multiply-connected domains~\cite{crowdy2020solving,nasser2020plgcirmap} for a broader range of shape analysis and geometry processing tasks. 

While we have demonstrated the effectiveness of our method for handling different surfaces, it is noteworthy that the density diffusion formulation assumes a smooth density. Therefore, a highly discontinuous and extreme density function induced by the input population may affect the accuracy of the final mapping result. Also, the current approach involves solving for an optimal Beltrami coefficient and reconstructing the entire circular domain at each iteration, which may be time-consuming for large datasets. With the recent development of parallelizable methods for conformal~\cite{choi2020parallelizable} and quasiconformal~\cite{zhu2022parallelizable} mapping, a natural next step is to explore the parallel computation of density-equalizing maps. Besides, our current work only focuses on connected open surfaces in $\mathbb R^3$. In the future, we plan to extend the method for closed surfaces and volumetric domains~\cite{choi2021volumetric,zhang2022unifying}.

\appendix
\section{Discretization and computation of density diffusion}\label{appendix:density}
Below, we describe the discretization and computation of density diffusion. Readers are referred to~\cite{choi2018density} for more details.

\subsection{Discretization of density}
Note that for every vertex $v \in \mathcal{V}$, the point density $\rho^{\mathcal{V}}(v)$ depends on the face density $\rho^{\mathcal{F}}$ of its 1-ring face neighborhood. Hence, the vertex density is given by 
\begin{equation}
    \rho^{\mathcal{V}}(v)=\frac{\sum_{\mathcal{T} \in \mathcal{N}^{\mathcal{F}}(v)} \rho^{\mathcal{F}}(\mathcal{T}) \operatorname{Area}(\mathcal{T})}{\sum_{\mathcal{T} \in \mathcal{N}^{\mathcal{F}}(v)} \operatorname{Area}(\mathcal{T})},
\end{equation}
where $\operatorname{Area}(\mathcal{T})$ is the area of the face $\mathcal{T}$ and $\mathcal{N}^{\mathcal{F}}(v)$ is the set of faces in the 1-ring neighborhood of the vertex $v$. If we regard $\rho^{\mathcal{V}}$ as a $|\mathcal{V}|\times 1$ matrix and $\rho^{\mathcal{F}}$ as a $|\mathcal{F}|\times 1$ matrix, the vertex density on the whole surface can be given by
\begin{equation}
    \rho^{\mathcal{V}} = W^{\mathcal{F} \mathcal{V}} \rho^{\mathcal{F}},
\end{equation}
where $W^{\mathcal{F} \mathcal{V}}$ is a $| \mathcal{F}| \times |\mathcal{V}  |$ sparse matrix:
\begin{equation}
    W^{\mathcal{F V}}=\left(
    \begin{array}{c}
    W_{1,:} /\left\|W_{1,:}\right\|_1 \\
    W_{2,:} /\left\|W_{2,:}\right\|_1 \\
    \vdots \\
    W_{|\mathcal{V}|:} /\left\|W_{|\mathcal{V}|,:}\right\|_1
    \end{array}\right),
\end{equation}
with 
\begin{equation}
    W_{i j}= \begin{cases}\operatorname{Area}\left(T_j\right) & \text { if the } j \text {-th triangle } T_j \text { contains the } i \text {-th vertex,} \\ 0 & \text { otherwise. }\end{cases}
\end{equation}

\subsection{Solving the diffusion equation}
To solve the diffusion equation~\eqref{diffusion_equation}, we discrete the Laplacian operator on the triangle mesh. Using the finite element method, the Laplacian operator $\Delta$ can be discretized as:
\begin{equation}
    \Delta \rho(v_i) = \frac{1}{A(v_i)} \sum_{j \in \mathcal{N}(v_i)} \frac{(\cot{\alpha_{ij}}+\cot{\beta_{ij}})}{2}(\rho(v_j) - \rho(v_i)),
\end{equation}
where $\mathcal{N}(v_i)$ is the 1-ring vertex neighborhood of $v_i$, $\alpha_{ij}$ and $\beta_{ij}$ are the two angles opposite to the edge $[v_i,v_j]$, and $A(v_i)$ is the vertex area with
\begin{equation}
A(v_i) = \frac{1}{3}\sum_{\mathcal{T}\in \mathcal{N}^{\mathcal{F}}(v_i)} \operatorname{Area}(\mathcal{T}),
\end{equation}
where $\mathcal{N}^{\mathcal{F}}(v_i)$ is the 1-ring face neighborhood of $v_i$. Alternatively, the Laplacian operator $\Delta$ can be represented as a matrix multiplication 
\begin{equation}
    \Delta =  A^{-1} L,
\end{equation}
where $A$ is a $|\mathcal{V}| \times |\mathcal{V}|$ diagonal matrix consisting of the vertex area and $L$ is a $|\mathcal{V}| \times |\mathcal{V}|$ sparse matrix with
\begin{equation}
    L_{i j}= \begin{cases} -\sum_{v_k \in \mathcal{N}(v_i)} \frac{(\cot{\alpha_{ik}}+\cot{\beta_{ik}})}{2} & \text { if } j = i, \\
    \frac{(\cot{\alpha_{ij}}+\cot{\beta_{ij}})}{2} & \text { if } v_j \in \mathcal{N}(v_i), \\ 
    0 & \text { otherwise. }\end{cases}
\end{equation}

With the densities defined on all vertices, the diffusion equation~\eqref{diffusion_equation} can be solved by the semidiscrete backward Euler method
\begin{equation}\label{discrete_diffusion}
    \frac{\rho^{\mathcal{V}}_{n} - \rho^{\mathcal{V}}_{n-1}}{\delta t} = \Delta_{n-1} \rho^{\mathcal{V}}_n,
\end{equation}
where $\rho^{\mathcal{V}}_{n}$ is the vertex density at $n$-th iteration, $\Delta_n$ is the Laplace operator, and $\delta t$ is the time step size. Combining with the Laplacian operator, the vertex density can be updated by
\begin{equation}
    \rho^{\mathcal{V}}_{n}  =(I- \delta t \Delta_{n-1})^{-1} \rho^{\mathcal{V}}_{n-1} = (A_{n-1} - \delta t L_{n-1})^{-1} (A_{n-1}\rho^{\mathcal{V}}_{n-1}).
\end{equation}

Then, we consider discretizing the gradient operator $\nabla$. 
Let $\mathcal{T} = [v_i, v_j, v_k]$ be a triangle element, where the three directed edges are denoted by $e_{ij} = [v_i,v_j]$, $e_{jk} = [v_j,v_k]$, and $e_{ki} = [v_k, v_i]$. Denote $N$ as the unit normal vector of $\mathcal{T}$. For any point $x$ lying on $\mathcal{T}$, $\rho^{\mathcal{F}}_{n}$ can be interpolated by
\begin{equation}
    \rho^{\mathcal{F}}_{n}(x) = \rho^{\mathcal{V}}_{n}(v_i) \varphi_i(x) + \rho^{\mathcal{V}}_{n}(v_j) \varphi_j(x) + \rho^{\mathcal{V}}_{n}(v_k) \varphi_k(x),
\end{equation}
where $\rho^{\mathcal{V}}_{n}(v_i)$, $\rho^{\mathcal{V}}_{n}(v_j)$, $\rho^{\mathcal{V}}_{n}(v_k)$ are vertex densities at $v_i, v_j, v_k$, and $\varphi_i$ are hat functions defined on the vertices:
\begin{equation}
    \varphi_i(p)= \begin{cases} 1 & \text { if } p = v_i,  \\
    \text{affine} & \text { if } p \in \mathcal{N}^\mathcal{F}(v_i), \\ 
    0 & \text { otherwise. }\end{cases}
\end{equation}
Using the property that $\nabla \varphi_i = \frac{N \times e_{j k}}{2\text{Area}(\mathcal{T})}$, $(\nabla \rho)^{\mathcal{F}}_{n} (\mathcal{T})$ can be expressed as
\begin{equation}\label{gradient}
\begin{split}
    (\nabla \rho)^{\mathcal{F}}_{n} (\mathcal{T}) & = \nabla \left(\rho^{\mathcal{V}}_{n}(v_i) \varphi_i(x) + \rho^{\mathcal{V}}_{n}(v_j) \varphi_j(x) + \rho^{\mathcal{V}}_{n}(v_k) \varphi_k(x)\right)\\
    & = \rho^{\mathcal{V}}_{n}(v_i) \nabla \varphi_i(x) + \rho^{\mathcal{V}}_{n}(v_j) \nabla \varphi_j(x) + \rho^{\mathcal{V}}_{n}(v_k) \nabla \varphi_k(x) \\
    & = \frac{1}{2\text{Area}(\mathcal{T})}N \times \left(\rho^{\mathcal{V}}_{n}(v_i) e_{jk} + \rho^{\mathcal{V}}_{n}(v_j) e_{ki} + \rho^{\mathcal{V}}_{n}(v_k) e_{ij}\right).
\end{split}
\end{equation}

After obtaining the gradient $(\nabla \rho)^{\mathcal{F}}_{n}$ on the triangle element, we can calculate the gradient operator on the vertices by
\begin{equation}\label{matrix:discrete}
    (\nabla \rho)^{\mathcal{V}}_{n} = W^{\mathcal{F} \mathcal{V}}  (\nabla \rho)^{\mathcal{F}}_{n}.
\end{equation}
As mentioned in the main text, we enforce the Neumann boundary condition $\nabla \rho \cdot \mathbf{n} = 0$ to preserve the circular shape of all boundaries. However, during the iterations, some numerical errors may occur, resulting in minor changes to the boundary shape as described in~\cite{choi2018density}. Hence, after every iteration, we add one step to project all boundary vertices onto the prescribed shape given by Algorithm~\ref{alg:GM}. As we have enforced the Neumann boundary condition, this projection step will not significantly affect the consistency between the density field and the shape deformation.

\section{Computing the boundary Beltrami coefficient in terms of the hole parameters}\label{appendix:BC}

In this section, we describe the process of computing the boundary Beltrami coefficient in terms of the parameters of an inner circular hole on the circular domain. Specifically, consider the $i$-th inner circular hole and a face $\mathcal{T}_l = [p_0, p_1, p_2]$ in the 1-ring face neighborhood of the hole. There are then two possibilities:
\begin{enumerate}[(i)]
    \item Exactly one vertex of $\mathcal{T}_l$ (say, $p_0$) lies on the circular hole.
    \item Exactly two vertices of  $\mathcal{T}_l$ (say, $p_0$ and $p_1$) lie on the circular hole.
\end{enumerate}
To simplify our discussion, we focus on the first case below and then briefly describe the second case at the end of this section. 

In the first case, we start by further simplifying some notations as follows. We denote the boundary Beltrami coefficient $\nu_l$ as $\nu$ and its corresponding triangle element $\mathcal{T}_l$ as $\mathcal{T}$. Additionally, we denote the hole parameters $(c_i,r_i,T_{i,j})$ for the boundary point $p_0$ as $(c,r,T)$.

Consider a piecewise linear function $g$. Let $w_k = g(p_k)$, where $p_k = m_k + in_k$ and $w_k = s_k + it_k$ for $k = 0, 1, 2$. On the triangle face $\mathcal{T}$, $g$ can be written as:
\begin{equation}
\left.g\right|_\mathcal{T}(x, y)=\left(\begin{array}{c}
a_\mathcal{T} x+b_\mathcal{T} y+p \\
c_\mathcal{T} x+d_\mathcal{T} y+q
\end{array}\right),
\end{equation}
where $D_xu(\mathcal{T}) = a_\mathcal{T}$, $D_yu(\mathcal{T}) = b_\mathcal{T}$, $D_xv(\mathcal{T}) = c_\mathcal{T}$, and $D_yv(\mathcal{T}) = d_\mathcal{T}$. 

Based on the Beltrami equation, the Beltrami coefficient $\nu(g)$ can be expressed as:
\begin{equation}
\nu(g) = \frac{\partial g / \partial \Bar{z}}{\partial g / \partial z} = \frac{(a_\mathcal{T} - d_\mathcal{T})+i(c_\mathcal{T} +b_\mathcal{T} )}{(a_\mathcal{T} + d_\mathcal{T})+i(c_\mathcal{T} - b_\mathcal{T} )}.
\end{equation}
For simplicity, we divide $\nu$ into two parts:
\begin{equation}\label{eqt:nu_two_parts}
\nu = \rho + i \tau = \frac{R}{Q} +i\frac{I}{Q}, 
\end{equation}
where 
\begin{equation}
\begin{aligned}
& R = a^2_\mathcal{T} - d^2_\mathcal{T} + c^2_\mathcal{T} - b^2_\mathcal{T}, \\
& I = 2a_\mathcal{T} b_\mathcal{T} + 2c_\mathcal{T} d_\mathcal{T}, \\
& Q = a^2_\mathcal{T} + d^2_\mathcal{T} + c^2_\mathcal{T} + b^2_\mathcal{T} 2a_\mathcal{T} d_\mathcal{T} + 2c_\mathcal{T} b_\mathcal{T}.
\end{aligned}
\end{equation}

The gradient of $g$ on every triangle face $\mathcal{T}$ is given by
\begin{equation}
\nabla_\mathcal{T} g = (D_xg(\mathcal{T}),D_yg(\mathcal{T}))^{T},
\end{equation}
where $D_xg(\mathcal{T}) = (D_x u, D_x v)$ and $D_yg(\mathcal{T}) = (D_y u, D_y v)$. To compute $\nabla_\mathcal{T} g$, we solve the following linear system:
\begin{equation}
\left(\begin{array}{c}
p_1-p_0 \\
p_2-p_0
\end{array}\right) \nabla_\mathcal{T} g=\left(\begin{array}{c}
\frac{g\left(p_1\right)-g\left(p_0\right)}{\left|p_1-p_0\right|} \\
\frac{g\left(p_2\right)-g\left(p_0\right)}{\left|p_2-p_0\right|}
\end{array}\right),
\end{equation}
where $[p_0,p_1]$ and $[p_0,p_2]$ are two edges of $\mathcal{T}$. After solving the above equation, the coefficients $a_\mathcal{T}$, $b_\mathcal{T}$, $c_\mathcal{T}$, and $d_\mathcal{T}$ can be obtained as follows:
\begin{equation}
\begin{aligned}
& a_\mathcal{T}=A_0 s_0 + A_1 s_1 + A_2 s_2, \\
& b_\mathcal{T}=B_0 s_0 + B_1 s_1 + B_2 s_2, \\
& c_\mathcal{T}=A_0 t_0 + A_1 t_1 + A_2 t_2, \\
& d_\mathcal{T}=B_0 t_0 + B_1 t_1 + B_2 t_2,
\end{aligned}
\end{equation}
where
\begin{equation}
\begin{aligned}
A_0 & =\left(n_1-n_2\right) / \operatorname{Area}(\mathcal{T}), \\
A_1 & =\left(n_2-n_0\right) / \operatorname{Area}(\mathcal{T}), \\
A_2 & =\left(n_0-n_1\right) / \operatorname{Area}(\mathcal{T}), \\
B_0 & =\left(m_2-m_1\right) / \operatorname{Area}(\mathcal{T}), \\
B_1 & =\left(m_0-m_2\right) / \operatorname{Area}(\mathcal{T}), \\
B_2 & =\left(m_1-m_0\right) / \operatorname{Area}(\mathcal{T}).
\end{aligned}
\end{equation}

With a focus solely on the boundary points, the formula above can be expressed as follows:
\begin{equation}\label{BC:express}
\begin{aligned}
& a_\mathcal{T}=\widetilde{A}_1+A_0 s_0, \\
& b_\mathcal{T}=\widetilde{B}_1+B_0 s_0, \\
& c_\mathcal{T}=\widetilde{A}_2+A_0 t_0, \\
& d_\mathcal{T}=\widetilde{B}_2+B_0 t_0.
\end{aligned}
\end{equation}
where $\widetilde{A}_1$, $\widetilde{B}_1$, $\widetilde{A}_2$, $\widetilde{B}_2$ are constants.

Substituting the above results into Eq.~\eqref{eqt:nu_two_parts}, the Beltrami coefficient at $p_0$ can be expressed as:
\begin{equation}
    \begin{split}
        \nu(s_0,t_0) = &\rho(s_0,t_0) + i\tau(s_0,t_0) \\ 
            = &\frac{R(s_0,t_0)}{Q(s_0,t_0)} + i\frac{I(s_0,t_0)}{Q(s_0,t_0)},
    \end{split}
\end{equation}
where
\begin{equation}
\begin{aligned}
& R(s_0,t_0) = (\widetilde{A}_1+A_0 s_0)^2 - (\widetilde{B}_2+B_0 t_0)^2 + (\widetilde{A}_2+A_0 t_0)^2 -(\widetilde{B}_1+B_0 s_0)^2 , \\
& I(s_0,t_0) = 2(\widetilde{A}_1+A_0 s_0)(\widetilde{B}_1+B_0 s_0)^2)+2(\widetilde{A}_2+A_0 t_0)(\widetilde{B}_2+B_0 t_0), \\
& Q(s_0,t_0) = (\widetilde{A}_1+A_0 s_0)^2 + (\widetilde{B}_2+B_0 t_0)^2 + (\widetilde{A}_2+A_0 t_0)^2 + (\widetilde{B}_1+B_0 s_0)^2 \\
 &  \quad \quad \quad \quad \quad \quad + 2(\widetilde{A}_1+A_0 s_0)(\widetilde{B}_2+B_0 t_0)  + 2(\widetilde{A}_2+A_0 t_0)(\widetilde{B}_1+B_0 s_0)^2.
\end{aligned}
\end{equation}
Using polar coordinates, we can express $s_0$ as $r \cos{\theta}-c^1$ and $t_0$ as $r \sin{\theta}-c^2$, where $c = (c^1,c^2)$ represents the center of the circular hole, $r$ denotes the radius of the circular hole, and $\theta$ denotes the angular coordinate of the point. Now, note that the angular component of the point can be expressed in terms of the tangent vector $T$ at it. Also, by an abuse of notation, $r$ is interchangeable with the radial vector as in the main text. Therefore, the boundary Beltrami coefficient $\nu_l$ associated with $\mathcal{T}_l$ can be expressed in terms of the hole parameters $(c,r,T)$:
\begin{equation}
\nu(s_0,t_0) = \nu(c,r,T) = \frac{R(c,r,T)}{Q(c,r,T)} + i\frac{I(c,r,T)}{Q(c,r,T)}.
\end{equation}

We now briefly discuss the second case where two vertices $p_0, p_1$ of the triangle element lie on the circular hole. Analogous to Eq.~\eqref{BC:express} in the first case, we can obtain the following expressions for $a_\mathcal{T}$, $b_\mathcal{T}$, $c_\mathcal{T}$, and $d_\mathcal{T}$:
\begin{equation}
\begin{aligned}
& a_\mathcal{T}=\hat{A}_1+A_0 s_0 +A_1 s_1, \\
& b_\mathcal{T}=\hat{B}_1+B_0 s_0 +B_1 s_1, \\
& c_\mathcal{T}=\hat{A}_2+A_0 t_0 +A_1 t_1 , \\
& d_\mathcal{T}=\hat{B}_2+B_0 t_0 + B_1 t_1,
\end{aligned}
\end{equation}
where $\hat{A}_1$, $\hat{B}_1$, $\hat{A}_2$, $\hat{B}_2$, $A_0$, $B_0$, $A_1$, and $B_1$ are some constants.

Using polar coordinates, we can express the coordinates of the two boundary points as follows: $s_0 = r \cos{\theta^0}-c^1$, $s_1 = r \cos{\theta^1}-c^1$, $t_0 = r \sin{\theta^0}-c^2$, and $t_1 = r \sin{\theta^1}-c^2$, where $c = (c^1,c^2)$ represents the center of the circular hole, $r$ denotes the radius of the circular hole, and $\theta^0$, $\theta^1$ denote the angular coordinates of the two points. Similar to the previous case, we can express $s_0, s_1, t_0, t_1$ using $(c, r^0, T^0)$ and $(c, r^1, T^1)$, where $r^0$, $r^1$ are two radial vectors and $T^0, T^1$ are two tangent vectors at the two points. We can then express $\nu_l$ in terms of $c, r^0, r^1, T^0, T^1$ as follows:
\begin{equation}
    \begin{split}
        \nu(s_0,t_0,s_1,t_1) = & \nu(c,r^0,r^1,T^0,T^1) \\
        = & \frac{R(c,r^0,r^1,T^0,T^1)}{Q(c,r^0,r^1,T^0,T^1)} + i\frac{I(c,r^0,r^1,T^0,T^1)}{Q(c,r^0,r^1,T^0,T^1)}.
    \end{split}
\end{equation}
Finally, note that in the domain alteration step of the geometry modification method, all points on an inner boundary are enforced to rescale and rotate consistently. Therefore, we can further simplify the above expression as $\nu = \nu(c, r, T)$ without ambiguity.

\bibliographystyle{ieeetr}
\bibliography{references}
\end{document}